\documentclass{article}
\usepackage[utf8]{inputenc}
\usepackage{amsmath}
\usepackage{amssymb}
\usepackage{amsthm}
\usepackage{mathtools}
\usepackage{xcolor}
\usepackage{comment}
\usepackage{subfigure}
\usepackage{authblk}

\newtheorem{theorem}{Theorem}
\newtheorem{proposition}{Proposition}
\newtheorem{lemma}{Lemma}
\newtheorem{assumption}{Assumption}
\newtheorem{remark}{Remark}
\newtheorem{definition}{Definition}
\usepackage{algorithm}
\usepackage{algpseudocode}

\usepackage{xcolor}

\usepackage[margin=3cm]{geometry}

\title{Towards practical PDMP sampling: Metropolis adjustments, locally adaptive step-sizes, and NUTS-based time lengths}

\author[1]{A. Chevallier}
\affil[1]{Université de Strasbourg}
\author[2]{S. Power}
\affil[2]{University of Bristol}
\author[3]{M. Sutton}
\affil[3]{Queensland University of Technology}

\begin{document}

\maketitle

\begin{abstract}
Piecewise-Deterministic Markov Processes (PDMPs) hold significant promise for sampling from complex probability distributions. However, their practical implementation is hindered by the need to compute model-specific bounds. Conversely, while Hamiltonian Monte Carlo (HMC) offers a generally efficient approach to sampling, its inability to adaptively tune step sizes impedes its performance when sampling complex distributions where scales differ dramatically across regions of the space.

To address these limitations, we introduce three innovative concepts: (a) a Metropolis-adjusted approximation for PDMP simulation that eliminates the need for explicit bounds without compromising the invariant measure, (b) an adaptive step size mechanism compatible with the Metropolis correction, and (c) a No U-Turn Sampler (NUTS)-inspired scheme for dynamically selecting path lengths in PDMPs. These three ideas can be seamlessly integrated into a single, `doubly-adaptive' PDMP sampler with favourable robustness and efficiency properties.
\end{abstract}

%%%%%%%%%%%%%%%%%%%%%%%%%%%%%%%%%%%%%%%%%%%%%%%%%%%%%%%%%%%%%%%%%%%%%%%%%%%%%%%%%%%%%%%%%%%%%%%%%%%%%%%%%%%%%%%%%%%%%
\section{Introduction}
Sampling from complex probability distributions is a fundamental task in Bayesian statistics. Piecewise Deterministic Markov Processes (PDMPs) have emerged as a promising class of methods for this task, with various compelling theoretical and practical features \cite{fearnhead2018piecewise,bouchard2018bouncy,bierkens2019zig}. However, despite their promise, the practical implementation of PDMP-based sampling methods has been hindered by the inherent computational complexities involved in simulating these processes. Existing implementations tend to rely on some combination of i) a priori analytical computations, ii) a priori quantitative knowledge of bounds and/or shape constraints on the target measure \cite{sutton2023concave}, and iii) the use of heuristic approximations, introducing a bias which is not necessarily addressed directly \cite{bertazzi2022approximations, bertazzi2301piecewise, corbella2022automatic, andral2024automated, pagani2024nuzz}.

Mainstream PDMP samplers are velocity-based algorithms, a class which also encompasses Hamiltonian Monte Carlo (HMC) \cite{neal2011mcmc}, Underdamped Langevin Monte Carlo (uLMC) \cite{cheng2018underdamped}, and others. These algorithms have broadly demonstrated their effectiveness in Bayesian sampling and make use of cleverly-designed kinetic dynamical systems as a means to navigate complex state spaces, whether deterministic, diffusive, piecewise-deterministic, or otherwise. A recurrent challenge for such methods is to design suitable protocols for updating the velocity variable over time, whether through discrete or continuous refreshment mechanisms. 

In the context of Hamiltonian Monte Carlo in particular, the No-U-Turn Sampler (NUTS) has emerged as a solution to this challenge. The core idea of the NUTS mechanism is to adaptively determine the integration time by terminating trajectories when they start to retrace their steps \cite{hoffman2014no}. While HMC-NUTS has been a substantial development \cite{Stan} for the practical adoption of HMC in the Bayesian community, the algorithm is not without limitations. A notable drawback is its lack of local adaptivity in choosing the step size, which can impact the efficiency of exploring target distributions, particularly those with spatial heterogeneity of conditioning. Indeed, HMC relies on a Metropolis correction which imposes that if a transition from $x$ to $y$ is possible, then the reverse transition from $y$ to $x$ must also be valid, which makes a naive local adaptation of the step-size impossible. 

However, recent advances have introduced strategies to enable adaptive step sizing while preserving the Metropolis correction. One approach incorporates the step size into the state space of the Markov chain itself \cite{bourabee2025gistgibbsselftuninglocally, bou2024incorporating, liu2024autostep}, while another leverages delayed rejection techniques \cite{Modi2024, modi2024atlas, turok2024sampling}. In both of these methods, the step size remains constant throughout a single HMC iteration but is adapted between iterations. The recently proposed Within Orbit Adaptive Leapfrog NUTS (WALNUTS) \cite{bourabee2025withinorbitadaptiveleapfrognouturn} takes a different approach. Rather than holding the step size constant, WALNUTS divides the time of a NUTS orbit into fixed macro steps. At each macro step, the algorithm tests a schedule of doubling and halving to find the largest micro step size that keeps the local energy error below a specific limit. This allows the sampler to use large steps in flat areas and small steps in curved areas during a single trajectory without violating detailed balance.

In this paper, we introduce three novel algorithms that make PDMP-based sampling both practical and robust. 
\begin{enumerate} 
\item Our first innovation is a `Metropolized' treatment of PDMPs, which enables the use of numerical approximations in PDMP simulation without introducing additional systematic bias in the ergodic properties of the sample. Moreover, our particular framing has the significant advantage of interfacing gracefully with local adaptivity in the selection of step-sizes and other computational hyperparameters.

\item The second innovation is to extend the use of the NUTS criterion to (exact) PDMPs, addressing the refreshment challenge specific to this class of processes. This method is not confined to the No-U-Turn criterion and can accommodate alternative criteria for adaptivity, making it a versatile addition to the PDMP sampling toolkit.

\item Finally, our third innovation combines the strengths of the previous developments into a comprehensive framework for the practical deployment of PDMP samplers, without requiring a priori knowledge of bounds or introducing bias, where \emph{both} the step size and the iteration-specific path length are adapted locally. This combined approach provides a robust and practical solution for a wide range of complex target distributions.
\end{enumerate}

We explored the empirical asymptotic scaling of the algorithms and compared them against standard HMC-NUTS on two challenging examples, observing improved robustness and performance. Notably, our experiments show that our algorithms outperform HMC, particularly in challenging scenarios such as the Funnel distribution proposed by Radford Neal \cite{neal2003slice}. In general, we expect our algorithms to excel (in relative terms) when dealing with complex state spaces characterized by multiple spatial lengthscales.

The paper is organized as follows. The next section introduces key concepts, including PDMP samplers, HMC, and the No-U-Turn Sampler. Section \ref{sec:metropolis} discusses the Metropolization of approximate PDMP samplers, while Section \ref{sec:nuts} presents a stopping criterion for determining path length in exact PDMP samplers. Section \ref{sec:full-algo} integrates these ideas, introducing a Metropolized PDMP approximation with locally adaptive step size and path length. Finally, Section \ref{sec:experiments} showcases numerical experiments and the paper concludes with a discussion in Section \ref{sec:disc}.

%However, recent advances have introduced strategies to enable adaptive step sizing while preserving the Metropolis correction. One approach incorporates the step size into the state space of the Markov chain itself \cite{bourabee2025gistgibbsselftuninglocally,bou2024incorporating,bourabee2025withinorbitadaptiveleapfrognouturn,liu2024autostep}, while another leverages delayed rejection techniques \cite{Modi2024,modi2024atlas,turok2024sampling}. In all of these methods, the step size remains constant throughout a single HMC iteration but is adapted between iterations, differing from the naive adaptation where the step size varies at each leapfrog step.

%In a series of empirical experiments, we demonstrate the superior robustness and performance of our algorithms in comparison to HMC.

%In summary, this paper introduces innovative PDMP-based sampling algorithms which take real steps torwards addressing several key challenges of PDMP simulation, namely i) local adaptivity of step size, and ii) refreshment policies for the velocity variables. Our empirical findings support the potential of these algorithms to significantly enhance the efficiency and robustness of Bayesian sampling across various scientific and computational domains.
\newpage

%%%%%%%%%%%%%%%%%%%%%%%%%%%%%%%%%%%%%%%%%%%%%%%%%%%%%%%%%%%%%%%%%%%%%%%%%%%%%%%%%%%%%%%%%%%%%%%%%%%%%%%%%%%%%%%%%%%%%
\section{Concepts}
\subsection{Piecewise-Deterministic Markov Processes}
\label{subsec:concepts-PDMP}
Piecewise-deterministic Markov processes (PDMPs) are a class of algorithms for Markov chain Monte Carlo (MCMC) sampling. They are defined as continuous-time stochastic processes which evolve in a piecewise-deterministic manner, in the sense that the process i) moves along deterministic trajectories until ii) a random event occurs, at which point iii) the process jumps to a new state.

Concretely, PDMPs are characterised by three natural properties: i) the driving deterministic motion, as described by a vector field $\phi$, ii) the event rate $\lambda$, which designates the local probability of a jump occurring at a given location, and iii) the jump kernel $Q$, which determines the law of where the process jumps to, given an event at a given location.

Practical simulation of PDMPs then reduces to i) solving the ODE corresponding to the deterministic motion, ii) simulating event times, and iii) simulating the jumps which occur at the jump times. For most PDMPs of practical interest, the process is designed so that i) and iii) are trivially solvable, and so the main challenges arise in solving ii), i.e. simulating event times. For `exact' implementations of PDMPs, this requires bounds on the value of the jump rate $\lambda$ along the path of the process, which typically requires a priori analytical knowledge about the size, shape, or other characteristics of the target measure. This represents a substantial bottleneck for the general-purpose adoption of PDMPs.

While a number of approximate solutions have been proposed \cite{bertazzi2022approximations, bertazzi2301piecewise, corbella2022automatic, andral2024automated, pagani2024nuzz} and have shown reasonable practical performance, in this work we focus on strategies which are \textit{asymptotically exact}, i.e. the invariant measure of the simulated process is precisely the probability measure of interest. We pursue this in the interests of simplicity and composability, noting that it i) sidesteps the thorny technical issue of bias control and estimation, and ii) enables the straightforward nesting of resulting algorithms within other Monte Carlo pipelines, mitigating the possibility of biases and errors propagating along a workflow.

\subsection{Kinetic Methods for MCMC Sampling}

Many `first-generation' methods for MCMC sampling on Euclidean spaces are based around the simple notion of simulating some physical motion of a `simple' particle $x \in \mathbb{R}^d$, stochastic, deterministic, or otherwise, so that the law of $x$ eventually tends to the measure of interest, $\pi$. 

A substantial practical innovation in subsequent generations has been to instead simulate the motion of a `kinetic' particle $\left( x, v \right) \in \mathbb{R}^d \times \mathbb{R}^d$, i.e. a particle with both a position and a velocity. In a suitable limit, one typically specifies the process such that the law of $\left( x, v \right)$ asymptotically tends to $\mu := \pi \otimes \psi$ for some simple law $\psi$, often e.g. a standard Gaussian. For kinetic systems, the dynamics are generally specified such that position $x$ moves in the direction of the velocity $v$. As such, the creative freedom in designing kinetic MCMC algorithms arises in how the velocity changes over time, whether through diffusion \cite{cheng2018underdamped}, through a combination of deterministic motion with discrete refreshment \cite{neal2011mcmc}, or other approaches. In various distinct settings, it has been observed that well-designed `kinetic' processes can have vastly improved convergence properties relative to their `simple' counterparts, both in theory and in practice. This has thus motivated substantial interest in the topic from all sides.

Hand-in-hand with this potential for improved convergence properties, a consistently-observed feature of various kinetic MCMC methods is that tuning of algorithmic hyperparameters requires particular care. In particular, for the sake of ergodicity, many such methods include some notion of `refreshment' of the velocity, wherein the position $x$ is held fixed, and the velocity $v$ is blended (either fully or partially) with a fresh sample from $\psi$. While simple in principle, the practitioner must decide both i) how frequently such refreshments should occur and ii) how aggressively these refreshments should alter the velocity.

\subsection{Hamiltonian Monte Carlo, and the No U-Turn Criterion}

The pre-eminent example of a kinetic MCMC procedure is arguably the Hamiltonian Monte Carlo (HMC) algorithm \cite{neal2011mcmc}. Given a kinetic particle $\left( x, v \right)$, HMC generates proposal moves by numerically simulating Hamiltonian dynamics with respect to the Hamiltonian function $H \left(x, v \right) = - \log \pi \left( x \right) + \frac{1}{2} \left| v \right|^2$, and then accepting or rejecting the proposed move with a suitable Metropolis adjustment step. These moves are interleaved with `full refreshments' of the velocity, i.e. drawing $v \sim \psi = \mathcal{N} \left( 0, I_d \right)$ independently of the current state. See \cite{neal2011mcmc} and \cite{betancourt2017conceptual} for a nice review.

From a tuning perspective, one key decision in HMC is how long of a dynamical trajectory to simulate before making the accept-reject decision. For too short a trajectory, the process loses the benefits of its kinetic character, sacrificing potential improvements in convergence behaviour. On the other hand, the simulation of longer trajectories incurs a higher computational cost and eventually admits diminishing returns with respect to mixing behavior. Some strategies have been proposed (e.g. \cite{wu2018faster, hoffman2021adaptive}) for tuning the trajectory length in an online fashion, but these typically focused on `global' strategies, i.e. using a single trajectory length, uniformly with respect to the starting position.

A substantial innovation was the No U-Turn Criterion \cite{hoffman2014no}, which proposed an automated procedure for adaptively determining a quasi-optimal trajectory length on an instance-by-instance basis. This enables the algorithm to develop long trajectories in challenging parts of the state space, while falling back to shorter trajectories in more regular regions. Roughly speaking, the algorithm proceeds by generating the Hamiltonian trajectory both forwards \textit{and} backwards in time from the initial state, growing it outwards until the trajectory begins to `double back' on itself, committing a `U-Turn' of sorts. Once this happens, the trajectory is fixed, and the next state is drawn from the full collection of states generated along the trajectory, according to a multinomial discrete distribution. The resulting MCMC algorithm is known as the No U-Turn Sampler (NUTS), a version of which is the basis of the MCMC engine of the Stan Probabilistic Programming Languages. By publicly presenting a well-engineered version of this automatically-adapted and self-tuning sampler, Stan has arguably enabled HMC to become the de facto algorithm of choice for general-purpose Bayesian posterior simulation in the current era.
\newpage

%%%%%%%%%%%%%%%%%%%%%%%%%%%%%%%%%%%%%%%%%%%%%%%%%%%%%%%%%%%%%%%%%%%%%%%%%%%%%%%%%%%%%%%%%%%%%%%%%%%%%%%%%%%%%%%%%%%%%
\section{Practical Simulation of PDMPs through Metropolis Adjustment}
\label{sec:metropolis}
In this section, we show how to use methods for approximate simulation of PDMPs to construct Markov chains of a given invariant measure. Our construction is grounded in a version of the Metropolis-Hastings procedure, whereby we use the approximate PDMP to propose a move through the state space, and then accept or reject this proposed move with a well-chosen `acceptance probability'. In order to maintain high acceptance rates, the proposal procedure must be designed to follow certain balance principles. In particular, if our approximation proposes to transition from the state $x$ to the state $y$, then the `matching' proposal from $y$ back to $x$ should also have appreciable likelihood. We ensure this by harnessing the \emph{skew-reversibility} property of several PDMPs, a property which arises from comparing the behaviour of the process as run forwards and backwards in time.

Section \ref{subsec:pdmp-reverse} defines what the \textit{time-reversal} process is for a given Markov process, and establishes a skew-reversibility property for the exact process. 
Section \ref{subsec:path-space} rigorously introduces path spaces for PDMP and defines relevant probability measures on said spaces.
Section \ref{subsec:approximation} details the construction of the PDMP approximation.
Section \ref{subsec:acceptance} then shows how to build the Metropolis adjustment for the approximated process.

\subsection{Time-reversal and skew-reversibility}
\label{subsec:pdmp-reverse}

In this section, we fix $T>0$ and consider a continuous-time Markov process $\{Z_t : 0 \leq t \leq T\}$. We first state a general result and then specialize to the case of piecewise deterministic Markov processes (PDMPs). The \emph{time-reversal} of $Z_t$ on $[0,T]$ is defined by
\begin{align*}
    Z^r : w \mapsto \bigl(t \mapsto w(T-t)\bigr), \qquad t \in [0,T].
\end{align*}
We denote by
\begin{align*}
    W := \{\, w : [0,T] \to \mathbb{R}^d \,\}
\end{align*}
the path space, i.e. the set of functions from $[0,T]$ to $\mathbb{R}^d$. For the reversed process we use the same underlying set, but write $W^r$ to avoid ambiguity. The law of $Z$ on $(W,\mathcal{F})$ (where $\mathcal{F}$ is the cylinder $\sigma$-algebra) is denoted by $P$, and the law of the reversed process $Z^r$ on $(W^r,\mathcal{F}^r)$ by $P^r$. We define the path-reversal map

\begin{align}
R : W &\to W^r, \\
   w  &\mapsto \bigl(t \mapsto w(T-t)\bigr).
\end{align}
By construction of the time-reversed process, for every measurable set $A \subset W$ one has
\begin{equation}
    P(A) \;=\; P^r(R(A)).
    \label{eq:backward}
\end{equation}
Using both the forward and the backward processes within a common Markov kernel allows one to define a reversible Markov chain.
\begin{proposition}[skew-reversibility]
    For a fixed $T > 0$, we define the Markov kernel $q$ on $E \times \{-1,1\}$ with the following steps starting from $(z,\gamma)$:
    \begin{enumerate}
        \item if $\gamma = 1$, follow the process forward in time starting from $z$ and return $(z_T,-\gamma)$,
        \item if $\gamma = -1$, follow the process backward in time starting from $z$ and return $(z^r_T,-\gamma)$.
    \end{enumerate}
    If the process admits $\mu$ as an invariant distribution, then $q$ satisfies detailed balance with respect to the product measure of $\mu$ and the uniform measure on $\{-1,1\}$.
    Consequently, $q$ leaves this product measure invariant.
\end{proposition}
\begin{proof}
    The probability $P_{\text{prod}}(A'\times B') = \int_{A'} \mu(z) q(B'|(z,\gamma)) \, \mathrm{d} z d\gamma $ defined as starting from $A'$ and ending in $B'$ defines a joint probability on the product space $E \times \{-1,1\} \times E \times \{-1,1\}$. The detailed balance condition can be stated as
    \begin{align*}
        P_{\text{\text{prod}}}(A' \times B' ) = P_{\text{prod}}( B'\times A')
    \end{align*}
    for any subsets $A',B'$ of $E \times \{-1,1\}$.
    
    Let $A$ and $B$ be subsets of $E$. Let $P$ and $P^r$ be the distributions on path spaces associated to the process at equilibrium, that is for any $t$, the law of $Z_t$ is $\mu$.
    First,
    \begin{align*}
        \int_A q(B\times\{\gamma\}|(z,\gamma)) \mu(z) \, \mathrm{d} z = 0 = \int_B q(A\times \{\gamma\}|(z,\gamma)) \mu(z) \, \mathrm{d} z,
    \end{align*}
    as $\gamma$ was not flipped. In other words, $P_{\text{prod}}(A\times \{\gamma\}\times B\times\{\gamma\}) = 0$, and so is the reverse. 
    Second, using \eqref{eq:backward} we get:
    \begin{align*}
        \int_A q(B\times \{-1\} |(z,1)) \mu(z) \, \mathrm{d} z &= P(\{w \in W | w_0 \in A, w_T \in B\}) \\
        &= P^r(\{w^r \in W^r | w^r_0 \in B, w^r_T \in A\}) \\
        &= \int_B q( A\times \{1\}|(z,-1)) \mu(z) \, \mathrm{d} z.
    \end{align*}
    In other word, $ P_{\text{prod}}(A\times\{1\} \times B\times\{-1\}) =  P_{\text{prod}}(B\times\{-1\}\times A\times\{1\})$. A similar reasoning can be made to prove that $ P_{\text{prod}}(A\times\{\gamma\} \times B\times\{-\gamma\}) =  P_{\text{prod}}(B\times\{-\gamma\}\times A\times\{\gamma\})$ for $\gamma \in \{-1,1\}$. Since any $A' = A_1\times\{1\} \cup A_2 \times \{-1\}$, and the same goes for $B'$,   $P_{\text{prod}}(A' \times B' ) = P_{\text{prod}}( B'\times A')$ for any $A',B'$ in $E \times \{-1,1\}$, which concludes the proof.
\end{proof}

This result is general; we now specialize it to PDMPs. We first describe the time-reversal of PDMPs, while subsequent sections provide further details on the associated path spaces.
It has been shown \cite{lopker2013} that the time-reversed process of a PDMP over $[0,T]$ is also a PDMP, whose characteristics can be identified more-or-less explicitly; we write $\lambda^r$ and $Q^r$ for its jump rate and jump kernel.  The general construction of the reverse process, that is $\lambda^r$ and $Q^r$, can be found in \cite{lopker2013}. We give here a special case that applies to both Bouncy Particle Sampler and the Zig-Zag Process.

\begin{proposition}[Practical time-reversal]
    Assuming that $Z_t$ is a PDMP with deterministic flow $\phi_t$, jump rate $\lambda$ and jump kernel $Q(z,\cdot) = \delta_{F(z)}$ where $F$ is an injection leaving the Lebesgue measure invariant, and that $\mu$ is an invariant distribution with a density, the time-reversal is the following PDMP:
    \begin{enumerate}
        \item Reverse deterministic flow, i.e. $\phi_{{rev}, t} = \phi_t^{-1}$
        \item Jump rate density $\lambda_{rev}(z) = \lambda(F^{-1}(z)) \frac{\mu(F^{-1}(z))}{\mu(z)} |J_{F^{-1}}(z)|$
        \item Jump Kernel $Q_{rev}(z,\cdot) = \delta_{F^{-1}(z)}$.
    \end{enumerate}
\end{proposition}

\subsection{Path Space and Path Measure for PDMPs}
\label{subsec:path-space}

The preceding discussion suggests a conceptual discrete-time algorithm for sampling from $\mu$: at each step, sample $\gamma$ uniformly from $\left\{ \pm 1 \right\}$ and then sample $\left( z^\prime, \gamma^\prime \right) \sim q \left( \left( z, \gamma \right), \cdot \right)$. This is readily seen to leave the product of $\mu$ with the uniform measure on $\{-1,1\}$ invariant, and is hence nominally fit for purpose as a sampler for $\mu$ (ignoring questions of efficiency). However, sampling from $q$ is often challenging, for reasons discussed in Section \ref{subsec:concepts-PDMP}. Nevertheless, if we are able to sample from some well-behaved approximation to $q$, we expect that a suitable Metropolis-adjustment procedure can deliver an algorithm which is both implementable and theoretically valid. In this Subsection, we elaborate on the specific notion of ``well-behavedness'' required for this correction: namely, the approximate proposal must admit a tractable density with respect to a suitable dominating measure. To establish this, it is crucial to first precisely define the path space of the process and the canonical reference measure with which it is equipped.

The path of a piecewise deterministic Markov process (PDMP) $(Z_t)_{t \in [0,T]}$ is fully determined by its \textit{skeleton}, which consists of the initial state $Z_0$, the event times $\tau_i$, and the state of the process immediately after each jump $Z_{\tau_i}$. Therefore, for a PDMP on state space $E$ with canonical measure $\lambda_E$, its path space can be expressed as:
\[
    E \times \bigcup_{k \in \mathbb{N}} \left( ]0,T[ \times E \right)^k,
\]
where $k$ is the number of jumps in a given path.
This space is endowed with the canonical product measure. However, if the jump kernel $Q$ does not possess a density with respect to the canonical measure of $E$, then the probability on $W$ induced by the PDMP does not have a density with respect to the canonical product measure. 

\begin{assumption}[Jump kernel parametrisation]
    \label{ass:jump-param}
    There exists a space $V$ with a canonical measure $\lambda_V$ such that for each $x \in E$, $Q( \cdot | x)$ is the pushforward of a probability density $\nu_x$ on $V$ by an injection $\xi_x : V \mapsto E$. In other words, drawing a sample from  $Q( \cdot | x)$ is equivalent to drawing a sample on $V$ from $\nu_x$ then applying $\xi_x$.
\end{assumption}
This assumption is true for the usual PDMP samplers such as BPS and the Zig-Zag Process. Under this assumption, the path space can be defined as
\begin{equation}
    W = E \times \bigcup_{k \in \mathbb{N}} \left( ]0,T[ \times V \right)^k.
\end{equation}
The canonical measure $\lambda_W$ on $W$ is defined as the product measure.

\begin{proposition}
\label{prop:path-density} 
\begin{enumerate}
    \item The distribution $P$ on the path space $W$, corresponding to the PDMP $(Z_t)_{t \in [0,T]}$, is absolutely continuous with respect to the canonical measure $\lambda_W$ on $W$.
    \item In particular, this implies that $P$ admits a density $p$ with respect to $\lambda_W$. Furthermore, for a given path, the density $ p((z_0,(\tau_1,v_1),...,(\tau_k,v_k)))$ is given by
\[
    \mu(z_0) e^{-\int_0^{T-\tau_k} \lambda(\Phi_t(z_k)) \, \mathrm{d} t} \prod_{i = 1}^k \lambda(\Phi_{\tau_{i} - \tau_{i-1}}(z_{i-1})) e^{-\int_0^{\tau_i - \tau_{i-1}} \lambda(\Phi_{t}(z_{i-1})) \, \mathrm{d} t} \nu_{\Phi_{\tau_{i} - \tau_{i-1}}(z_{i-1})}(v_i),
\]
    where $z_i = \xi_{\Phi_{\tau_{i} - \tau_{i-1}}(z_{i-1})}(v_i)$, that is the state of the PDMP after the $i$-th jump.
    \item The probability of a path $z$ conditionally on its starting point $z_0$ has a density $p(z | z_0)$ on $\{z_0\} \times \bigcup_{k \in \mathbb{N}} \left( ]0,T[ \times V \right)^k$.
\end{enumerate}

\end{proposition}
\begin{proof}
This can be seen by explicitly constructing the probability distribution on $W$ from the skeleton of the process.
Let us write a typical element of $W$ as $(z_0, \tau_1, v_1, \dots, \tau_k, v_k)$, where:
\begin{itemize}
    \item $z_0 \sim \mu$, where $\mu$ is a distribution on $E$ with a density with respect to the base measure $\lambda_E$. Hence, the initial state $z_0$ has a density with respect to $\lambda_E$,
    \item the event times $\tau_i$ follow (conditionally on $\tau_{i-1}$ and $z_{i-1}$) the first event time of an inhomogeneous Poisson process, with rate $\lambda(\Phi_{t-\tau_{i-1}}(z_{i-1}))$, and
    \item the post-jump state $z_i$ is determined by $v_i$, which has density $\nu_{\Phi_{\tau_{i} - \tau_{i-1}}(z_{i-1})}$ on $V$.
\end{itemize}
Thus, the joint distribution of $(z_0, \tau_1, v_1, \dots, \tau_k, v_k)$ is given by a product of these densities: the density of $z_0$ with respect to $\lambda_E$, the conditional densities of the jump times $\tau_i$, and the densities $\nu$ with respect to $\lambda_V$. This shows that the probability distribution $P$ on $W$ has a density $p$ with respect to the canonical measure $\lambda_W$ on $W$.
\end{proof}

\subsection{Approximation of a PDMP}
\label{subsec:approximation}

Exactly simulating the path of a PDMP, while possible in various cases, generally requires computing target-dependent upper bounds for $\lambda$. We build here a numerical approximation $\tilde{Z_t}$ to circumvent this problem. In particular, we consider approximations which i) can be simulated easily, and ii) whose density with respect to the reference measure on path space is readily available. We then simulate the corresponding process for a fixed time $T$. By construction, we can compare this proposed path to its time-reversal (whose density is similarly available), and then use this as the basis for a Metropolis-adjustment of the proposal $\widetilde{Z_T}$.%

This process is very similar to HMC in general flavour, but differs in the nature of approximation which is made. Indeed, the discretisation introduced by HMC effectively prevents any kind of local adaptation of the step-size, as for many natural strategies, the adaptation has an intrinsically `forward-in-time' character, which renders the forward process singular with respect to to its time reversal (though see \cite{nishimura2016variable} for a counter-proposal along these lines). Here, we aim for a locally-adaptive approximation to the event rate function. To that end, the process $Z_t$ is approximated by \textit{another continuous-time} process $\tilde{Z_t}$. To do that, we build a \textit{continuous-time} approximation $\tilde{\lambda}(t)$ of $\lambda(t) = \lambda (\Phi_t(z_0), \gamma)$. Starting from $(z_0,\gamma)$ (where $\gamma \in \left\{ \pm 1 \right\}$ is the time direction), the approximated process is therefore simulated as follows:
\begin{itemize}
    \item While $\tau_i < T$:
    \begin{enumerate}
    \item at event time $\tau_i$, compute the approximation $\tilde{\lambda}(t)$ of $\lambda(\Phi_{\gamma t}(\tilde{z}_{\tau_i}),\gamma)$,
    \item simulate the next event time $\tau_{i+1} - \tau_{i}$ \textit{exactly} from the first event time of an inhomogeneous Poisson process with rate $\tilde\lambda(t)$,
    \item simulate $\tilde z_{\tau_{i+1}} \sim Q(\cdot | \Phi_{\gamma(\tau_{i+1} - \tau_i)}(\tilde z_{\tau_i}),\gamma)$,
\end{enumerate}
\end{itemize}
until time $T$ is reached, and then return $(\tilde z_T,\gamma)$. It should be noted that when viewed as a continuous-time stochastic process, the approximated process no longer has the Markov property, as the implied event rate depends on where the last jump occurred. The skeleton of the approximated process can still be a Markov process, but it would require extra assumptions on the way the approximation $\tilde{\lambda}$ is built. In any case, from the point of view of the Metropolis-adjustment, the Markov property of the proposal is not important; the interior of the trajectory is essentially just an auxiliary instrument for simulating the terminal point $\tilde z_T$.

The path space is the same for the exact process and the approximation, since the same jump kernel $Q$ is used. As such, we can write an analogue of Proposition \ref{prop:path-density} for the approximate process.
\begin{proposition}
\label{prop:path-density-approx}
\begin{enumerate}
    \item The distribution $\tilde{P}$ on the path space $W$, corresponding to the approximate PDMP $(\tilde{Z}_t)_{t \in [0,T]}$, is absolutely continuous with respect to the canonical measure $\lambda_W$ on $W$.
    \item In particular, $\tilde{P}$ admits a density $\tilde{p}$ with respect to $\lambda_W$, and for a given path, $ \tilde{p}((z_0,(\tau_1,v_1),...,(\tau_k,v_k)))$ can be written as:
\[
    \mu(z_0) e^{-\int_0^{T-\tau_k} \tilde{\lambda}(\Phi_t(z_k)) \, \mathrm{d} t} \prod_{i = 1}^k \tilde{\lambda}(\Phi_{\tau_{i} - \tau_{i-1}}(z_{i-1})) e^{-\int_0^{\tau_i - \tau_{i-1}} \tilde{\lambda}(\Phi_{t}(z_{i-1})) \, \mathrm{d} t} \nu_{\Phi_{\tau_{i} - \tau_{i-1}}(z_{i-1})}(v_i),
\]
    where $z_i = \xi_{\Phi_{\tau_{i} - \tau_{i-1}}(z_{i-1})}(v_i)$, that is the state of the PDMP after the $i$-th jump.
    \item The probability of a path $z$ conditionally on its starting point $z_0$ has a density $\tilde{p}(z | z_0)$ on $\{z_0\} \times \bigcup_{k \in \mathbb{N}} \left( ]0,T[ \times V \right)^k$.
\end{enumerate}

\end{proposition}
\begin{proof}
    Similar to Proposition \ref{prop:path-density}.
\end{proof}

The previous proposition, which will be a core part of the Metropolis correction, doesn't require any specific assumption on $\tilde \lambda$. However, we add two assumptions on $\tilde\lambda$ to ensure it can be properly simulated and admit a workable density: 

\begin{assumption}
    \label{ass:approx}
    \begin{enumerate}
        %\item $\tilde\lambda$ does not depend on what happens before $(z_0,\gamma)$,
        \item At any time $t \in [\tau_i, \tau_{i+1})$, the approximate rate $\tilde{\lambda}(t)$ is a measurable function of the initial state $(z_0, \gamma)$, the sequence of past jump times and locations $(\tau_1, z_1, \dots, \tau_i, z_i)$, and the elapsed time. In other words, it does not depend on the state of the process prior to $(z_0, \gamma)$,
        \item The first event time of an inhomogeneous Poisson process with rate $\tilde \lambda$ can be sampled exactly,
        \item The integral $\int_0^t \tilde \lambda(t) \, \mathrm{d} t$ can be computed exactly.
    \end{enumerate}    
\end{assumption}
The first condition ensures that the kernel that associates $(\tilde z_0,\gamma)$ to $(\tilde z_T,\gamma)$ is a Markov kernel, the second condition ensures that we can simulate the process in practice, and the third condition allows for the Metropolis adjustment.

\paragraph{Piecewise-Constant Rate Approximation}
The simplest approximation is for a fixed step size $h$:
\[
    \tilde\lambda^{(0)}_h(t) = \lambda(\left\lfloor t\right\rfloor_h),
\]
with $\left\lfloor t\right\rfloor_h = h \cdot \left\lfloor \frac{t}{h}\right\rfloor \in \left( t - h, t \right]$.

\paragraph{Piecewise-Linear Rate Approximation}
Another natural approximation is to take
\[
    \tilde\lambda^{(1)}_h(t) = \lambda(\left\lfloor t\right\rfloor_h) + \frac{t -\left\lfloor t\right\rfloor_h}{h} \cdot \left( \lambda(\left\lfloor t\right\rfloor_h + h) - \lambda(\left\lfloor t\right\rfloor_h)\right),
\]
which is piecewise-linear in time.
Observe that while this approximation might appear to `cheat' by looking ahead to time $\left\lfloor t\right\rfloor_h + h$, it satisfies assumption \ref{ass:approx}. The approximating rate itself is a simple, deterministic function of the initial value of $z$, and so does not cause any troubles in terms of yielding a tractable density for the proposal. 

\begin{remark}
    This approximation is exact if $\lambda$ is linear. This covers the case of Gaussian targets for both the BPS and the Zig-Zag Process.
\end{remark}

\subsection{Local adaptation of the step size}
Ideally, one would like the algorithm the algorithm to be \textit{scale-invariant}, that is, that applying the same algorithm to the scaled targets $\left\{ \sigma^{-d} \cdot \pi \left( \sigma^{-1} \cdot x \right) \right\}_{\sigma > 0}$ should yield the same effective result, irrespective of the scaling factor $\sigma$. One sees readily that any fixed-$h$ implementation of this procedure will fail to have this property; an idealised invariant implementation would take $h \propto \sigma$. Bearing this in mind, we seek implementation strategies which can automatically emulate this self-scaling behaviour.

\paragraph{Piecewise-Constant Rate Approximation}
Here, we present a rule for the order-0 approximation that satisfies Assumption \ref{ass:approx}, which implies it will be compatible with the Metropolis correction. 
This rule corresponds to the standard local adaptation of the step size for the Euler integrator. First, recall that simulating an event time for the process amounts to solving for some $u$
\[
    I_t = u,
\]
with $I_t = \int_0^t \lambda(s) ds$. Using a Taylor series expansion at $t = 0$, argue that
\[
    I_t = \lambda(0) t + K t^2 + \mathcal{O}(h_{guess}^3).
\]
For $t < h$, the first term corresponds to the piecewise constant approximation, while the second term is an error term. Assuming that $K$ is known, to bound the error on $I_t$ by a given tolerance, we therefore chose 
\[
    h = \sqrt{\frac{\text{tol}}{K}}.
\]
To find an estimate of $K$ for a given initial step size $h_{guess}$, we build two approximations of $I_{h_{guess}}$ using left Riemann sums: one step of size $h_{guess}$ ($A_1$) and two steps of size $h_{guess}/2$ ($A_2$). 
Comparing these to the Taylor expansion of the integral gives:
\begin{align*}
    A_1 &= I_{h_{guess}} - K h_{guess}^2 + \mathcal{O}(h_{guess}^3) \\
    A_2 &= I_{h_{guess}} - K \frac{h_{guess}^2}{2} + \mathcal{O}(h_{guess}^3)
\end{align*}
Subtracting $A_1$ from $A_2$ yields $A_2 - A_1 = K \frac{h_{guess}^2}{2} + \mathcal{O}(h_{guess}^3)$. We therefore estimate $K$ as $\hat{K} = 2 \frac{A_2 - A_1}{h_{guess}^2}$, which satisfies:
\begin{equation}
    \label{eq:K_approx}
    \hat{K} = K + \mathcal{O}(h_{guess}).
\end{equation}
The final adaptation rule is given by
\[
    h = h_{guess} \sqrt{\frac{\text{tol}}{2|\tau|}},
\]
where the difference $\tau = A_2 - A_1$ is computed as:
\begin{align*}
    \tau &= \lambda(0) \frac{h_{guess}}{2} + \lambda\left(\frac{h_{guess}}{2}\right) \frac{h_{guess}}{2} - \lambda(0) h_{guess} \\
         &= \frac{h_{guess}}{2} \left[ \lambda\left(\frac{h_{guess}}{2}\right) - \lambda(0) \right].
\end{align*}

\begin{proposition}[Scale-Invariance]

    Let $\lambda(t)$ be the twice continuously differentiable rate function for BPS sampler for an unscaled target $\pi$. Let $h_{guess}$ be the initial step size, and $h_\sigma$ the adapted step size computed for the scaled target $\pi_\sigma(x) = \pi(\sigma x)$ at $(\frac{1}{\sigma} x,v)$. As $h_{guess} \to 0^+$, it holds that
    \[
        h_\sigma = \frac{h_{guess}}{\sigma} \cdot  \left( 1 + \mathcal{O}(h_{guess}) \right),
    \]
    where $h$ is the adapted step size computed for the unscaled target.
\end{proposition}

\begin{proof}
    Starting from $(\frac{1}{\sigma}x,v)$, for BPS sampler the rate for the scaled target is $\lambda_\sigma(t) = \sigma \lambda(t\sigma)$. 
    Following the previous notations, let $K_\sigma$ be the second order coefficient of the Taylor expansion of $I_{t,\sigma} = \int_0^t \lambda_\sigma(s)ds$. Since $K_\sigma = \frac{1}{2}\lambda'_\sigma(0)$, then
    \[
        K_\sigma = \frac{1}{2} \sigma^2 \lambda'(0) = \sigma^2 K,
    \]
    where $K$ is the coefficient for the unscaled target $\pi$.
    Applying Equation \ref{eq:K_approx} to the scaled process, the algorithm computes an empirical estimate $\hat{K}_\sigma$ satisfying:
    \[
        \hat{K}_\sigma = K_\sigma + \mathcal{O}(h_{guess}) = \sigma^2 K + \mathcal{O}(h_{guess}).
    \]
    Finally, the unscaled empirical step size satisfies $h = \sqrt{\frac{\text{tol}}{|\hat{K}|}} = \sqrt{\frac{\text{tol}}{|K|}}(1 + \mathcal{O}(h_{guess}))$. Substituting $\hat{K}_\sigma$ into the scaled step size formula $h_\sigma = \sqrt{\frac{\text{tol}}{|\hat{K}_\sigma|}}$ yields:
    \[
        h_\sigma = \sqrt{\frac{\text{tol}}{|\sigma^2 K + \mathcal{O}(h_{guess})|}} = \frac{1}{\sigma} \sqrt{\frac{\text{tol}}{|K|}} \left( 1 + \mathcal{O}(h_{guess}) \right) = \frac{1}{\sigma} h \left( 1 + \mathcal{O}(h_{guess}) \right).
    \]
\end{proof}

\begin{remark}
    The naive adaptation which tries to find a $h$ such that $|\lambda(t) - \lambda(t+h)| < \text{tol}$ is not scale-invariant. 
\end{remark}

\paragraph{Piecewise-Linear Rate Approximation}

A similar adaptation rule can be derived for the order-1 approximation using a higher-order Taylor expansion. For brevity, we state only the final rule here and defer the complete derivation to Appendix \ref{appendix:order-1-adaptation}. Given an initial step size guess, $h_{guess}$, the updated step size is:

\begin{align*}
h &= h_{guess}\left(\frac{3\times \text{tol}}{|h_{guess}f(h) -2h_{guess}f(h_{guess}/2)+hf(0)|}\right)^{1/3}.
\end{align*}

\begin{proposition}[Scale invariance - order 1]
    \label{prop:scale-invariance-order-1}
    Let $\lambda(t)$ be the twice continuously differentiable rate function for BPS sampler for an unscaled target $\pi$. Let $h_{guess}$ be the initial step size, and $h_\sigma$ the adapted step size for the order 1 approximation computed for the scaled target $\pi_\sigma(x) = \pi(\sigma x)$ at $(\frac{1}{\sigma} x,v)$. As $h_{guess} \to 0^+$,
    \[
        h_\sigma = \frac{h_{guess}}{\sigma} \cdot \left( 1 + \mathcal{O}(h_{guess}) \right),
    \]
    where $h$ is the adapted step size for the order 1 approximation computed for the unscaled target.
\end{proposition}
\begin{proof}
    See Appendix \ref{appendix:order-1-adaptation}.
\end{proof}

\subsection{Metropolis-Adjustment of PDMP Discretisations}
\label{subsec:acceptance}

We assume a similar assumption to Assumption \ref{ass:jump-param} for the backward process.
\begin{assumption}[Reverse jump kernel parametrisation]
    There exists a space $V^r$ with a canonical measure $\lambda_{V^r}$ such that for each $x \in E$, $Q^r( \cdot | x)$ is the pushforward of a probability density $\nu^r_x$ on $V$ by an injection $\xi^r_x : V^r \mapsto E$. In other words, drawing a sample from  $Q^r( \cdot | x)$ is equivalent to drawing a sample on $V^r$ from $\nu^r_x$ then applying $\xi^r_x$.
\end{assumption}
Therefore, the path space for the backward process is
\begin{equation}
    W^r = E \times \bigcup_{k \in \mathbb{N}} \left( ]0,T[ \times V^r \right)^k.
\end{equation}
The approximation can be then built both forward and backward in time, with the backward approximation having probability $\tilde{P^r}$ on $W^r$ with density $\tilde{p^r}$.
To compute the Metropolis correction, the change of volume induced by $R$ must be taken into account.
\begin{assumption}[Volume change induced by R]
    We assume there exists a function $\psi: W^r \mapsto \mathbb{R}^+$  such that for any $f: W^r \mapsto \mathbb{R}$:
    \[
        \int_W f(R(w)) \lambda_W(\, \mathrm{d} w) = \int_{W^r} f(w^r) \psi(w^r) \lambda_{W^r}(\, \mathrm{d} w^r).
    \]
\end{assumption}
We may write \eqref{eq:backward} as 
\begin{equation}
    p(w) = p^r(R(w)) \frac{1}{\psi(R(w))}.
    \label{eq:backward-density}
\end{equation}
We identify $R^{-1}: W^r \mapsto W$ to $R$ for notation convenience. Hence $R$ is seen as a function mapping $W \cup W^r$ to itself, and the associated volume change is therefore $1/\psi(R(w))$ for $w \in W$. 

\begin{proposition}[Volume change for BPS and ZZP]
    For BPS and the Zig-Zag Process, there exists a PDMP representation (i.e. $V$ and $V^r$) such that $\psi$ exists and:
    \[
        \psi(z) = 1 \quad \forall z.
    \]
\end{proposition}
\begin{proof}
See Section \ref{subsec:ZZ-BPS}.
\end{proof}
The Metropolis-adjusted PDMP is defined in Algorithm \ref{algo:mh-approx} as a Markov kernel on $E \times \{-1,1\}$.
\begin{algorithm}[H]
    \caption{Metropolis adjusted PDMP Markov kernel}
    \label{algo:mh-approx}
    \begin{algorithmic}[1]
        \State \textbf{Input:} state $\tilde z_0,\gamma$
        \State Sample a path: 
        \begin{enumerate}
            \item forward approximated path $(\tilde w_t)_{t\in [0,T]}$ with $\tilde{w}_0 = \tilde{z_0}$ if $\gamma = 1$,
            \item or a backward approximated path $(\tilde w^r_t)_{t\in [0,T]}$ with $\tilde{w}^r_0 = \tilde{z_0}$ if $\gamma = -1$
        \end{enumerate}
    \State Compute the acceptance probability 
    \begin{enumerate}
        \item $\alpha(\tilde w) = 1 \wedge\frac{\mu(\tilde w_T) \tilde{p}^r(R(\tilde w) | \tilde{w_T}) }{\mu(\tilde{z}_0) \tilde{p}(\tilde w|\tilde{z}_0) \psi(R(\tilde{w}))}$ if $\gamma = 1$
        \item $\alpha(\tilde w^r) = 1 \wedge\frac{\mu(\tilde w^r_T) \tilde{p}(R(\tilde w^r) | \tilde{w^r_T}) \psi(\tilde{w}^r)}{\mu(\tilde{z}_0) \tilde{p}^r(\tilde w^r|\tilde{z}_0) }$ if $\gamma = -1$
    \end{enumerate}
    \State With probability $\alpha$, return the point $\tilde w_T,-\gamma$ if $\gamma = 1$ or $\tilde w^r_T,-\gamma$ if $\gamma = -1$
    \State Otherwise return $\tilde z_0,\gamma$  
    \end{algorithmic}
\end{algorithm}

\begin{theorem}[Metropolis-corrected PDMP]

\begin{enumerate}
    \item The Metropolis-adjusted PDMP Markov kernel verifies detailed balance for the product between $\mu$ and a uniform probability on the space $E \times \{-1,1\}$.
    \item If the `approximated process' used is the exact process, then the acceptance probability is 1.
\end{enumerate}
\end{theorem}
\begin{proof}
    For the first claim, write $\tilde{q}$ the Markov kernel.
    To prove detailed balance, it is enough to prove that for any $A,B \subset E$:
    \begin{equation}
        \int_A \mu(z) \tilde{q}(B\times\{-1\}|z) \, \mathrm{d} z = \int_B \mu(z) \tilde{q}^r(A\times\{1\}|z) \, \mathrm{d} z.
        \label{eq:proof-detailed-balance-approx}
    \end{equation}
    We have that
    \[
        \int_A \mu(z) \tilde{q}(B\times\{-1\}|z) \, \mathrm{d} z = \int_W \mu(w_0) \tilde{p}(w|w_0) 1_{A}(w_0) 1_B(w_T) \alpha(w) \, \mathrm{d} w
    \]
    and
    \begin{align*}
        \int_B \mu(z) \tilde{q}^r(A\times\{1\}|z) \, \mathrm{d} z &= \int_{W^r} \mu(w^r_0) \tilde{p}^r(w^r|w^r_0) 1_{A}(w^r_T) 1_B(w^r_0) \alpha(w^r) \, \mathrm{d} w^r \\
        &= \int_{W} \mu(R(w)_0) \tilde{p}^r(R(w)|R(w)_0) 1_{A}(w_0) 1_B(w_T) \alpha(R(w)) \frac{1}{\psi(R(w))} \, \mathrm{d} w^r.
    \end{align*}
    Equation \ref{eq:proof-detailed-balance-approx} is therefore equivalent to
    \[
        \mu(w_0) \tilde{p}(w|w_0) \alpha(w) = \mu(R(w)_0) \tilde{p}^r(R(w)|R(w)_0) \alpha(R(w)) \frac{1}{\psi(R(w))},
    \]
    which is clearly true. For the second claim, check that Equation \eqref{eq:backward-density}  implies that $\alpha = 1$ in the exact case.
\end{proof}

\subsection{BPS and Zig-Zag}
\label{subsec:ZZ-BPS}
We show here how to apply the previous results to BPS and the Zig-Zag Process. In particular, the path spaces for both of these process will be made explicit, the reversal application $R$ will be described and it will be shown to preserve volumes in both cases.

\subsubsection{BPS}

Here, $E = \mathbb{R}^d \times \{ v \in \mathbb{R}^d \text{ such that } \|v\| = 1 \}$. For $z \in E$, we write $z = (x,v)$ with $x \in \mathbb{R}^d$ and $v \in \{ v \in \mathbb{R}^d \text{ such that } \|v\| = 1 \}$.
For BPS, jumps are a reflection against $\nabla \pi$, which is deterministic. Hence, the jump kernel does not need to be parameterized, and therefore 
\[
    V = \{ 0 \}
\] 
is a singleton (where $0$ is meaningless). The jump is then defined as
\begin{align*}
    \xi_{(x,v)}: V &\rightarrow E \\
    0 &\mapsto \left(x,v - 2 (\nabla \pi(x) \cdot v) \frac{\nabla \pi(x)}{\|\nabla \pi(x)\|^2}\right).
\end{align*}
The reverse process is similar, and 
\begin{align*}
    V^r &= V, \\
    \xi^r_z &= \xi_z^{-1}.
\end{align*}
The application $R$ can be then written as
\[
    R(z_0,t_0,v_0,...,t_k,v_k) = R(z_T,T - t_k,v_k, ..., T - t_0,v_0).
\]
Furthermore, for a fixed $t_0,v_0,...,t_k,v_k$, the application $z_0 \mapsto z_T$ preserves the volume as it can be written as a composition of measure preserving operations, since the maps $\xi_z$ and the deterministic flow each preserve the volume.
Therefore, $R$ preserves the volume, i.e. 
\[
    \psi(w^r) = 1.
\]
For similar reasons, we can also see that $R$ has the same regularity as $\nabla \pi$, i.e. if $\nabla \pi$ is $C^k$, so is $R$.

\subsubsection{The Zig-Zag Process}

Here, $E = \mathbb{R} \times \{-1,1\}^d$. For $z \in E$, we write $z = (x,v)$ with $x \in \mathbb{R}^d$ and $v \in \{-1,1\}^d$.
Unlike BPS, the jump kernel is not deterministic. It can be characterized however by which velocity is flipped. Hence
\[
    V = \{1,...,d\}.
\]
The jump is then defined as
\begin{align*}
    \xi_{(x,v)}: V &\rightarrow E \\
    i &\mapsto (x,v_1,...,v_{i-1}, -v_i,v_{i+1},...,v_d).
\end{align*}
The reverse process is similar and
\begin{align*}
    V^r &= V, \\
    \xi^r_z &= \xi_z.
\end{align*}
The application $R$ can be then written as:
\[
    R(z_0,t_0,v_0,...,t_k,v_k) = R(z_T,T - t_k,v_k, ..., T - t_0,v_0),
\]
since the same coordinate is flipped at each event in the backward representation as in the forward representation.
Furthermore, for a fixed $t_0,v_0,...,t_k,v_k$, the application $z_0 \mapsto z_T$ preserves the volume as it can be written as a composition of measure preserving operations, since the maps $\xi_z$ and the deterministic flow preserve the volume.
Therefore, $R$ preserves the volume, i.e. 
\[
    \psi(w^r) = 1.
\]
Since the deterministic flow is $C^\infty$, and so are the jumps $\xi_z$, we can see that $R$ is $C^\infty$.

\newpage

%%%%%%%%%%%%%%%%%%%%%%%%%%%%%%%%%%%%%%%%%%%%%%%%%%%%%%%%%%%%%%%%%%%%%%%%%%%%%%%%%%%%%%%%%%%%%%%%%%%%%%%%%%%%%%%%%%%%%
\section{Refreshing PDMP velocities using a stopping criterion}
\label{sec:nuts}
PDMP samplers, like most velocity-driven samplers, require some form of velocity refreshment in order to support ergodicity. Determining the optimal refreshment rate is a significant challenge, even in rather simple cases. Moreover, the optimal refreshment rate may depend on the position within the state space, rendering a single, fixed parameter inefficient across different regions.

To address this challenge, we introduce an algorithm inspired by the NUTS criterion, as used in HMC \cite{hoffman2014no}, adapted for PDMP samplers. This algorithm constructs a path until a stopping criterion is satisfied, with minimal assumptions regarding the criterion. We propose a modified version of the No-U-Turn stopping criterion that meets these assumptions. 

\subsection{Stopping criterion}

In this context, a stopping criterion is a function that takes a path as input (of finite but otherwise arbitrary length), and outputs whether the path is valid, where `validity' reflects that it is still worthwhile to further extend the trajectory. For example, the (simplified) No-U-Turn Criterion checks if the path has a U-turn inside \cite{hoffman2014no}, in the sense that it verifies that for every points $(x_{t_1},v_{t_1})$ and $(x_{t_2},v_{t_2})$ of the trajectory with $t_1 < t_2$,
\begin{align*}
    \frac{\mathrm{d} \|x_{t_1} - x_{t_2}\|^2}{\mathrm{d}t_2} &\geq 0, \\
    \frac{\mathrm{d} \|x_{t_1} - x_{t_2}\|^2}{\mathrm{d}t_1} &\leq 0.
\end{align*}
The first inequality indicates that when building the trajectory forward in time, it should not reduce the distance with any other point of the trajectory, while the second indicates the same when building the trajectory backward in time. The two conditions can be rewritten as
\begin{align*}
    (x_{t_2} - x_{t_1}) \cdot v_2 \geq 0, \\
    (x_{t_2} - x_{t_1}) \cdot v_1 \geq 0.
\end{align*}

We call a stopping criterion \textit{consistent} if it satisfies the following assumption:
\begin{assumption}[consistent stopping criterion]
\label{ass:consistent}
Let $(X_s)_{s\in \mathbb{R}}$ be a ``full" path. Then:
\begin{enumerate}
    \item if $(X_s)_{s\in I}$ is valid for an interval $I$, then $(X_s)_{s\in J}$ is valid for any interval $J \subset I$,
    \item if $(X_s)_{s\in I}$ is not valid for an interval $I$, then $(X_s)_{s\in J}$ is not valid for any interval $J \supset I$.
\end{enumerate}
\end{assumption}
Furthermore, we require that:
\begin{assumption}
    \label{ass:criterion-events}
    The validity of the criterion for a path \((X_s)_{s\in \mathbb{R}}\) depends only on the points where events occur.
\end{assumption}
This implies that, for example, when extending an interval \([0,t]\) until the criterion is no longer valid, the process will always stop precisely at an event.

The criterion used in experiments is inspired by the No-U-Turn criterion introduced for HMC in NUTS \cite{hoffman2014no}, but has significant differences in view of \ref{ass:criterion-events}. Perhaps the key difference is that while the NUTS criterion depends on the full simulated Hamiltonian trajectory, our proposed criterion depends only on the state of the trajectory at the jump times which are simulated as part of PDMP.

\begin{definition}[PDMP No-U-Turn criterion]
    Let $Z = \left( Z_t : t \in  \mathbf{R} \right)$ be a PDMP path which experiences events (i.e. velocity jumps) at times $T = \left\{ t_i : 1 \leq i \leq N \right\} \subseteq \mathbf{R}$. For $a < b$, write $T_{\left[ a, b\right]} = T \cap \left[ a, b\right]$. We then say that $Z$ satisfies the PDMP No-U-Turn Criterion on $\left[ a, b\right]$ if
    \begin{align*}
        &\forall t_i < t_j \in [a,b], (X_{t_i} - X_{t_j},V_{t_j^-}) < 0, \\
        &\forall t_i < t_j \in [a,b] \text{ and } t_j < b, (X_{t_i} - X_{t_j},V_{t_j^+}) < 0, \\
        &\forall t_i > t_j \in [a,b], (X_{t_i} - X_{t_j},V_{t_j^+}) > 0, \\
        &\forall t_i > t_j \in [a,b], \text{ and } t_i > a, (X_{t_i} - X_{t_j},V_{t_j^-}) > 0.
    \end{align*}
\end{definition}

\subsection{Algorithm}

We consider a criterion that follow Assumptions \ref{ass:consistent} and \ref{ass:criterion-events}. The following algorithm is an adaptation to continuous time of the idea behind NUTS \cite{hoffman2014no}.
\begin{algorithm}[H]
    \caption{No-U-Turn PDMP sampler}
    \label{algo:nuts-exact}
    \begin{algorithmic}[1]
        \State \textbf{Input:} state $z=(x,v)$
        \State Refresh the velocity
        \State Sample $\alpha \sim unif(0,1)$
        \State Sampling a path:
        \begin{enumerate}
            \item Build a trajectory $Z_t$ for $t\in \mathbb{R}$ with $Z_0 = z$ (both forward and backward)
            \item Let $T = \max \{t\in \mathbb{R}^+ \, | \, Criterion((Z_s)_{s\in [-\alpha t, (1-\alpha)t]}) \text{ is valid}\} $
            \item Let $X = (Z_{t - \alpha T})_{t \in [0,T]}$ and $l = \alpha T$
        \end{enumerate}
        \If{the criterion was hit going forward}
            \State Sample $l' \sim \nu$, where $\nu$ is a probability measure on $[0,T]$ with density $\nu(l) = \frac{2(T-t)}{T^2}$.
        \Else
            \State Sample $l' \sim \nu$, where $\nu$ is a probability measure on $[0,T]$ with density $\nu(l) = \frac{2t}{T^2}$.
        \EndIf
        \State \textbf{Output:} $X_{l'}$        
    \end{algorithmic}
\end{algorithm}

Note that while this algorithm involves continuous-time simulation as a part of its inner workings, the Markov chain which it generates is ultimately a discrete-time process. The PDMP acts as an instrumental proposal which we expect to provide useful spatial exploration. 

\begin{theorem}
    Assuming a PDMP with invariant distribution $\mu$, such that the reversal $R_t$ and its volume change $\Psi_t$ are continuous in $t$. It then holds that Algorithm \ref{algo:nuts-exact} defines a Markov chain with $\mu$ as an invariant distribution.
\end{theorem}
\begin{proof}
The theorem is a direct consequence of the following Proposition \ref{prop:law_PlX}.    
\end{proof}

\begin{proposition}
    \label{prop:law_PlX}
    Assume that $Z_t$ is a PDMP with invariant distribution $\mu$, such that the reversal $R_t$ and its volume change $\Psi_t$ are continuous in $t$. Let $X$ be constructible (i.e. could be obtained by the algorithm).
    
    If the law of the initial state of the algorithm is $\mu$, then the law $\mathbb{P}(l | X)$ is $\nu$ where $\nu$ is a probability measure on $[0,T]$ with density $\nu(t) = \frac{2(T-t)}{T^2}$ if the criterion was hit going forward, and $\nu(t) = \frac{2t}{T^2}$ if the criterion was hit backward.
\end{proposition}

\begin{proof}
We give here the intuition of the proof; Appendix \ref{sec:nuts-proof} provides the rigorous details. In principle, when looking at a path, it is impossible to distinguish whether the process was built forward in time or backward in time. If the process is at equilibrium, it is also impossible to know where the process started from. On this basis, one expects that the conditional law of $l|X$ should then be uniform. 

However, this perspective misses one piece of the puzzle. Let $(X,l)$ be a path built by the algorithm, noting that the path length $T$ can be directly deduced from the information included in $X$. Assumption \ref{ass:criterion-events} tells us that the path was stopped at an event at one of the extremities. Without loss of generality, we assume that it was stopped by a backward event. To this $(X,l)$ corresponds a unique $\alpha$ in the construction, which can be computed $\alpha = l/T$. Let $a = -\alpha T$ and $b = (1-\alpha)T$. Since the path is stopped by an event, a small modification of $\alpha$ will not change $a$.  However, for $\alpha$ in the neighbourhood, $b(\alpha) = a - a/\alpha$. Therefore 
\[
    b'(\alpha) = -\frac{a}{\alpha^2}=\frac{l}{(l/T)^2} = \frac{T^2}{l}.
\]
The associated change of volume is then readily computed as $\frac{1}{b'}$, which gives the announced formula.
\end{proof}

\newpage

%%%%%%%%%%%%%%%%%%%%%%%%%%%%%%%%%%%%%%%%%%%%%%%%%%%%%%%%%%%%%%%%%%%%%%%%%%%%%%%%%%%%%%%%%%%%%%%%%%%%%%%%%%%%%%%%%%%%%
\section{Doubly-Adaptive PDMP Samplers}
\label{sec:full-algo}
It is now possible to combine the two previous ideas into a single algorithm, whereby we use both i) an approximate PDMP with a locally adaptive time step for the approximation, and ii) a stopping criterion for a locally adaptive path length. 
The idea is to use an approximated process in the algorithm introduced in Section \ref{sec:nuts}. 

\begin{proposition}
    \label{prop:law-l-approximated}
    Assume that the exact process is a PDMP such that the reversal $R_t$ and its volume change $\Psi_t$ are continuous in $t$.
    Using the same algorithm to build $(Z,l)$ as in Section  \ref{sec:nuts} but using an approximate process yields the following $\nu$ on $[0,T]$ for the law of $l | Z$:
    \[
        \nu(l) \propto \mu(Z_l) q^r_{[-l,0]}((Z_{t+l})_{t\in [-l,0]}|Z_l) q_{[0,T-l]}((Z_{s+l})_{s\in [0, T-l]}|Z_l) \left| \frac{T-l}{T^2} \right| \Psi_l(Z_{[0,l]}),
    \]
    if the process was stopped forward, or
    \[
        \nu(l) \propto \mu(Z_l) q^r_{[-l,0]}((Z_{t+l})_{t\in [-l,0]}|Z_l) q_{[0,T-l]}((Z_{t+l})_{t\in [T-l,T]}| Z_l) \left| \frac{l}{T^2} \right| \Psi_l(Z_{[0,l]}),
    \]
    if the process was stopped backward.
\end{proposition}
\begin{proof}
    Note that the state space and time reversal of the approximate process is the same as that of the exact process.
    The proof follows the proof of Appendix \ref{sec:nuts-proof}. The only difference between the proof for the exact process of Appendix \ref{sec:nuts-proof} and this proof is that the final simplification of Proposition \ref{prop:nuts-proof-law-l} cannot be done for the approximated process, which leads to the above formulas.
\end{proof}
\begin{remark}
    The assumption that $\mu$ is invariant for the exact continuous-time PDMP is not required.
\end{remark}
\begin{algorithm}[H]
    \caption{Doubly adaptive PDMP sampler}
    \label{algo:doubly-adaptive}
    \begin{algorithmic}[1]
        \State \textbf{Input:} state $z=(x,v)$
        \State Refresh the velocity
        \State Sample $\alpha \sim unif(0,1)$
        \State Sampling a path:
        \begin{enumerate}
            \item Build a trajectory (potentially adaptively) $Z_t$ for $t\in \mathbb{R}$ with $Z_0 = z$ (both forward and backward)
            \item Let $T = \max \{t\in \mathbb{R}^+ \, | \, Criterion((Z_s)_{s\in [-\alpha t, (1-\alpha)t]}) \text{ is valid}\} $
            \item Let $X = (Z_{t - \alpha T})_{t \in [0,T]}$ and $l = \alpha T$
        \end{enumerate}
        \If{the criterion was hit going forward}
            \State Sample $l' \sim \nu$, where $\nu$ is a probability measure on $[0,T]$ with density $\nu(l) = \frac{2(T-t)}{T^2}$.
        \Else
            \State Sample $l' \sim \nu$, where $\nu$ is a probability measure on $[0,T]$ with density $\nu(l) = \frac{2t}{T^2}$.
        \EndIf
        \State Accept $l'$ with probability 
        \[
            \frac{\mu(X_{l'})}{\mu(X_l)}\frac{ q^r_{[-l,0]}((X_{t+l'}')_{t\in [-l',0]}|X_{l'}) q_{[0,T-l']}((X_{t+l'})_{t\in [T-l',T]}|X_{l'}) \Psi_{l'}(X_{[0,l']})}{q^r_{[-l,0]}((X_{t+l})_{t\in [-l,0]}|X_l) q_{[0,T-l]}((X_{t+l})_{t\in [T-l,T]}|X_l) \Psi_l(X_{[0,l]})},
        \] 
        otherwise set $l' = l$.
        \State \textbf{Output:} $X_{l'}$        
    \end{algorithmic}
\end{algorithm}

\begin{remark}[BPS and Zig-Zag]
    For BPS and Zig-Zag, $\Psi_l \equiv 1$.
\end{remark}

\begin{theorem}
    Assume that the exact process is a PDMP such that the reversal $R_t$ and its volume change $\Psi_t$ are continuous in $t$.
    \begin{enumerate}
        \item Algorithm \ref{algo:doubly-adaptive} defines a Markov chain with $\mu$ as an invariant distribution.
        \item Furthermore, if $\mu$ is the invariant distribution of the exact process and if the approximated process is the exact process, the acceptance probability is 1.
    \end{enumerate}
\end{theorem}
\begin{proof}
    This is a direct consequence of Proposition \ref{prop:law-l-approximated}.
\end{proof}

% \begin{comment}
% \begin{enumerate}
%     \item Refresh the velocity
%     \item Sample $\alpha \sim unif(0,1)$
%     \item Sampling a path:
%     \begin{enumerate}
%         \item Build a trajectory $Z_t$ for $t\in \mathbb{R}$ with $Z_0 = z$ (both forward and backward)
%         \item Let $T = \max \{t\in \mathbb{R}^+ \, | \, Criterion((Z_s)_{s\in [-\alpha t, (1-\alpha)t]}) \text{is valid}\} $
%         \item Let $X = (Z_{t + \alpha T})_{t \in [0,T]}$ and $l = \alpha T$
%     \end{enumerate}
%     \item Sample $l' \sim \nu$, where $\nu$ is a probability measure on $[0,T]$ with density 
%     \[
%         \nu(l) \propto \pi(X_l) q^r_{[-l,0]}((X_{t+l})_{t\in [-l,0]}|X_l) q_{[0,T-l]}((X_{t+l})_{t\in [T-l,T]}|X_l) \left| \frac{T-l}{T^2} \right| \Psi_l(X_{[0,l]}),
%     \]
%      if the criterion was hit going forward, and 
%      \[
%         \nu(l) \propto \pi(X_l) q^r_{[-l,0]}((X_{t+l})_{t\in [-l,0]}|X_l) q_{[0,T-l]}((X_{t+l})_{t\in [T-l,T]}|X_l) \left| \frac{l}{T^2} \right| \Psi_l(X_{[0,l]}),
%     \]
%      if the criterion was hit backward.
%     \item output: $X_{l'}$
% \end{enumerate}
% \end{comment}
\newpage

%%%%%%%%%%%%%%%%%%%%%%%%%%%%%%%%%%%%%%%%%%%%%%%%%%%%%%%%%%%%%%%%%%%%%%%%%%%%%%%%%%%%%%%%%%%%%%%%%%%%%%%%%%%%%%%%%%%%%
\section{Experiments}
\label{sec:experiments}
This section presents our experimental results, utilizing the Bouncy Particle Sampler given its prevalence among PDMP samplers.
The experiments are divided into three main parts. Section \ref{subsec:exp-dim} analyzes the dimensionality scaling of our algorithms and compares it to the exact PDMP in some stylised settings for which the ideal PDMP can be simulated exactly. We demonstrate that in terms of the computational complexity required to obtain an independent sample in high dimension, there is no substantial loss of efficiency relative to the exact PDMP.
Section \ref{subsec:exp-funnel} tests our algorithms on a funnel distribution, a case where exact PDMP simulation is difficult. These experiments show the practical value of adjusting step size and path length locally. Section \ref{subsec:simulated-tempering} then evaluates our method in a simulated tempering setup. In both cases, our approach performs better than Hamiltonian Monte Carlo (HMC). A key takeaway from both sections is that our algorithm reliably samples the distribution tails, whereas HMC struggles to do so.

\subsection{Dimensionality scaling}
\label{subsec:exp-dim}

In this section, we study the sampling of the $d$-dimensional multivariate normal distribution $\mathcal{N}(0_d,I_d)$. We use the Effective Sample Size (ESS) on the first coordinate as a quality metric of the generated samples. For exact PDMPs, the complexity is the number of events, while for approximated PDMPs, the complexity is the number of gradient evaluations.

\subsubsection{Exact BPS-NUTS}

In the case of a multivariate normal distribution $\mathcal{N}(0_d,I_d)$, BPS is expected to require $O(\sqrt{d})$ events per independent sample (see \cite{deligiannidis2021scaling} for a detailed analysis). We expect the $O(\sqrt{d})$ scaling to hold for the NUTS version of exact BPS. 

\begin{remark}
    \label{rmk:exp-order1-gaussian}
   For Gaussian targets, the order 1 approximation of the rate is exact. Hence we studied the scaling properties of exact BPS-NUTS using the order 1 rate approximation using a large step size. 
\end{remark}

Figure \ref{fig:order1-scaling} shows the behavior of exact BPS using the No-U-Turn stopping criterion. A few key observations can be made.
First, Plots (a) and (b) indicate that the predicted $O(\sqrt{d})$ scaling is accurate in this case. The ESS per NUTS-step is roughly constant with the dimension and the number of PDMP events simulated per NUTS step is roughly scaling as $O(\sqrt{d})$. 

Second, as noted in Remark \ref{rmk:exp-order1-gaussian}, the order 1 approximation is used instead of the exact process. We can then see on (c) and (d) of Figure 1, that the overhead introduced by computing the approximate process (which is exact in this case) is between 7 and 8, that is each event requires around 7-8 gradient evaluations. Our code is not optimized, and we believe it should be possible in general to reduce this number by an additional factor of three or so. 

\begin{figure}
\centering
    \subfigure[]{\includegraphics[width=0.39\linewidth]{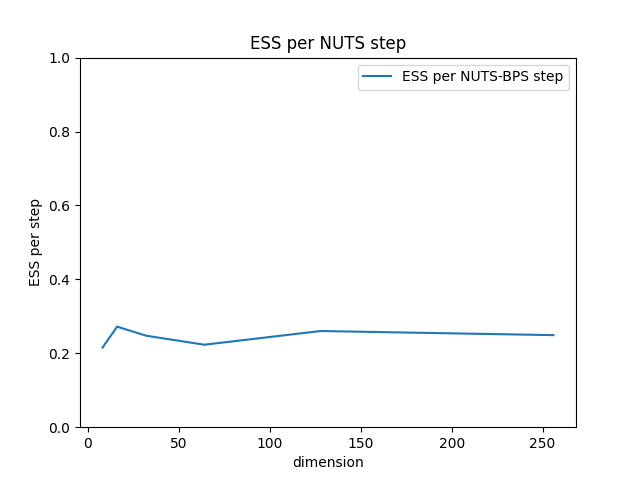}}
    \subfigure[]{\includegraphics[width=0.39\linewidth]{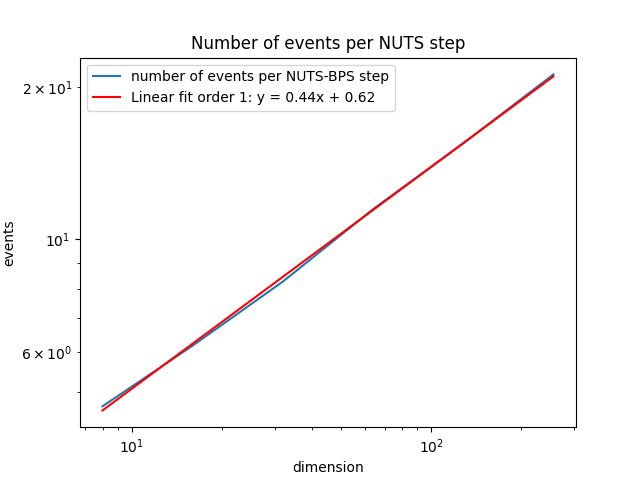}}
    \subfigure[]{\includegraphics[width=0.39\linewidth]{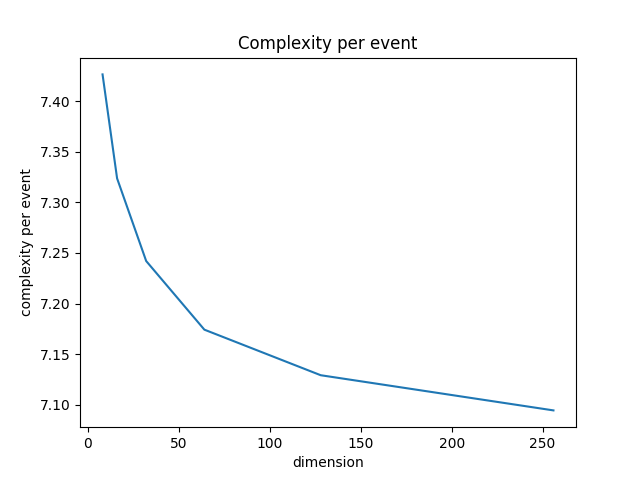}}
    \subfigure[]{\includegraphics[width=0.39\linewidth]{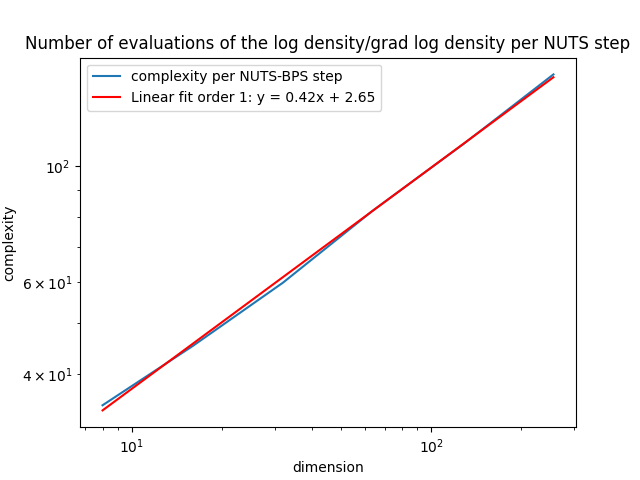}}
    \caption{Gaussian targets, order 1 approximation, 1000 NUTS-BPS steps, for various dimensions.}
    \label{fig:order1-scaling}
\end{figure}

\subsubsection{BPS-NUTS with no step-size adaptivity}

The aim of this experiment is to investigate the impact of the step size \( h \) in the approximation, using the No-U-Turn criterion. We assume a piecewise constant approximation of the rate function (order 0) so our BPS approximation is not exact. Specifically, we seek to determine:
\begin{enumerate}
    \item Whether the optimal value of \( h \) varies with the dimension \( d \).
    \item Whether the acceptance rate for the optimal \( h \) changes with \( d \).
    \item Whether the complexity scaling remains \( O(\sqrt{d}) \).
\end{enumerate}

In our BPS implementation, we have chosen the velocity \( v \) to have a norm of 1. This implies that at stationarity, the component of \( v \) in any specific direction decreases with the dimension as \( O(1/\sqrt{d}) \). The optimal step size \( h \) is determined using a grid search for each dimension, with the objective being to maximize the ESS in the first dimension divided by the computational complexity.

\begin{figure}
\centering
    \subfigure[]{\includegraphics[width=0.29\linewidth]{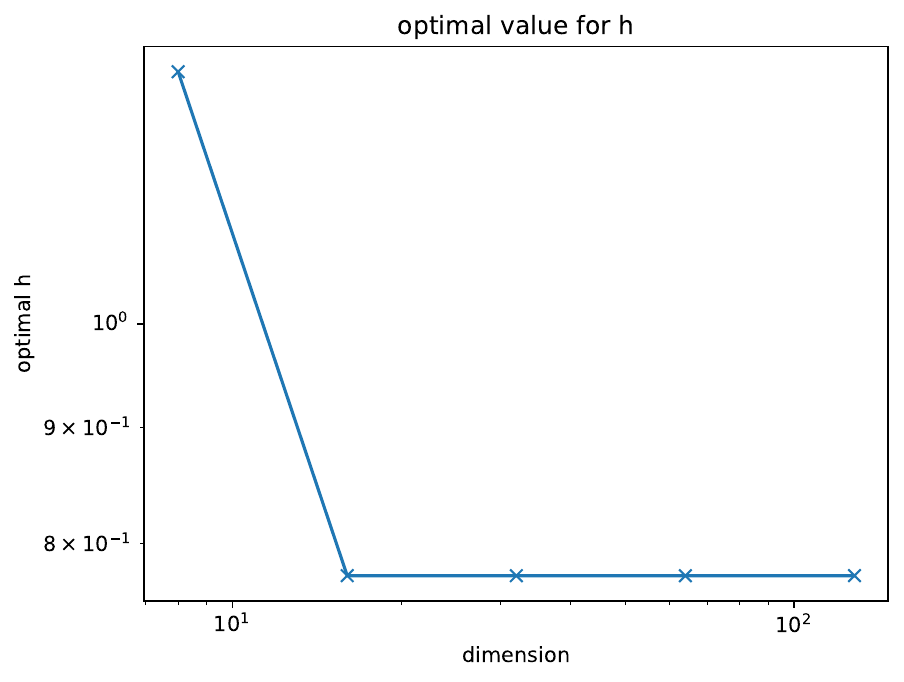}}
    \subfigure[]{\includegraphics[width=0.29\linewidth]{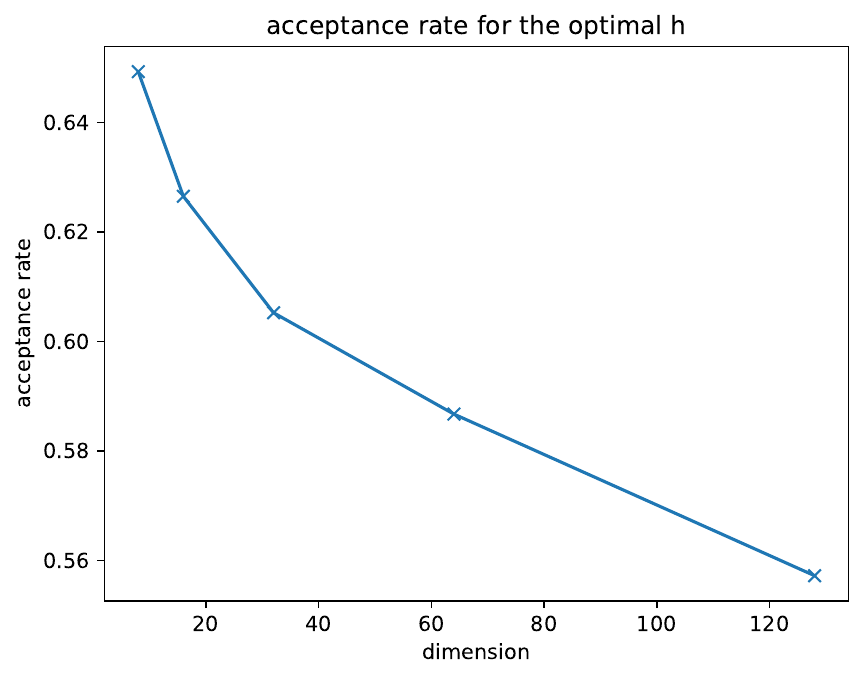}}
    \subfigure[]{\includegraphics[width=0.29\linewidth]{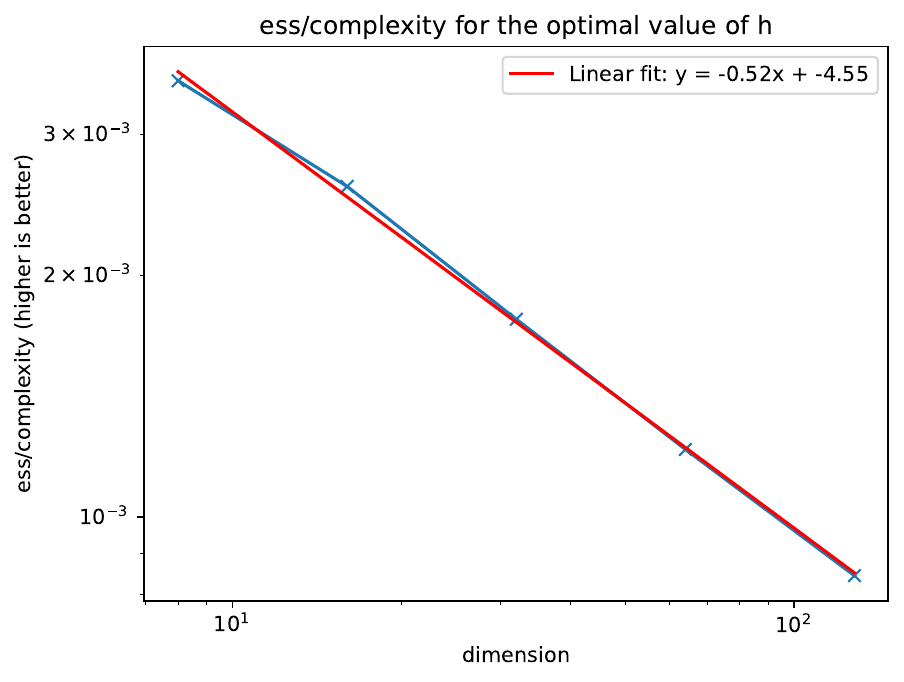}}
    \caption{(a) Optimal scaling for the stepsize $h$. (b) Acceptance rate for the optimal stepsize $h$. (c) Scaling of the ESS/complexity for the optimal value of $h$.}
    \label{fig:h-scaling}
\end{figure}
Figure \ref{fig:h-scaling}(a) clearly indicates that the value of \( h \) should not be scaled with dimension.

Moreover, the acceptance rate does not appear to converge to a non-zero value when using this optimal \( h \), as shown in Figure \ref{fig:h-scaling}(b). An optimal acceptance rate converging to \( 0 \) would represent a significant departure from the usual behavior of MCMC algorithms, such as Metropolis or HMC. Nevertheless, Figure \ref{fig:h-scaling} does not provide a definitive answer to this question, as the acceptance rate may converge to a non-zero value in higher dimensions.

Finally, Figure \ref{fig:h-scaling}(c) shows that the scaling of the complexity per effective sample for BPS-NUTS is $O( \sqrt{d})$ which matches the optimal BPS computational scaling for sampling a Gaussian distribution \cite{deligiannidis2021scaling}.

\subsection{Funnel distribution}
\label{subsec:exp-funnel}

In this experiment, we show that our algorithms perform robustly when applied to the task of sampling from the funnel distribution, in the sense that without substantial hand-tuning, the algorithm is numerically stable and efficiently converges to the target distribution at a reasonable computational cost.

The so-called funnel distribution is a toy problem that illustrates a common pathology in Bayesian hierarchical modeling \cite{neal2003slice}. The posteriors in these models are characterized by a region of low volume with high density which fans out to a region of high volume and low density. This presents a challenging sampling problem as the sampler will need to adapt across the space. Traditional adaptation in HMC can be overly assertive \cite{betancourt2015hamiltonian} and it is sometimes recommended to run a secondary chain at a smaller step size to ensure consistent inference. 

We study the 2d funnel $x_1 \sim \mathcal{N}(0,a^2)$ and $x_2 \sim \mathcal{N}(0,e^{\frac{x_1}{b}})$ parameterized by $a$ and $b$ in $\mathbb{R}$.
First, we showcase the effect of the adaptation on the step size in Figure \ref{fig:funnel-traj-ex}, for $a = 3$ and $b = 2$. As the sampler reaches a narrower section of the funnel the gradient evaluations get closer together. Due to the adaptive adjustment of $h$ in the sampler, the forward (blue) and backward (red) trajectories no longer share identical gradient evaluations.  
\begin{figure}[!h]
    \centering
    \includegraphics[width=0.5\linewidth]{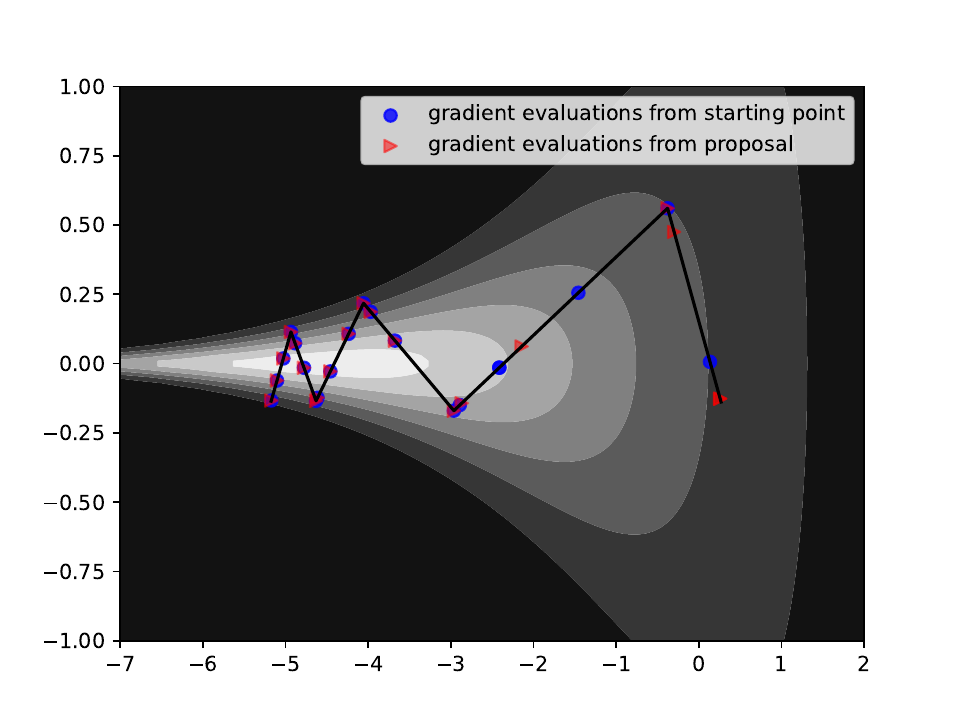}
    \caption{An example trajectory in a funnel for MHBPS-NUTS, with the location of gradient evaluations displayed in blue and red. The effect of the adaptivity on the stepsize can clearly be seen.}
    \label{fig:funnel-traj-ex}
\end{figure}

Second, we compare the performance of our algorithm with double adaptation to HMC with various step sizes in Figure \ref{fig:funnel-comparison}, using \( a = 3 \) and \( b = 1.5 \). Performance is assessed by an error metric computed using the marginal distribution of the first coordinate $X_1$, where we partition the space into three regions $X_1<-4, -4\leq X_1 \leq 4, X_1>4$. The error is defined to be the maximum of the discrepancy between the empirical and exact probabilities of these regions on the log scale. The figure shows results of 20 independent repetitions of each sampler, with each run consisting of 10,000 iterations. Due to the continual adaption of $h$ in our samplers, our algorithm outperforms HMC, which clearly struggles to sample the funnel.

\begin{figure}[!h]
    \centering
    \includegraphics[width=0.5\linewidth]{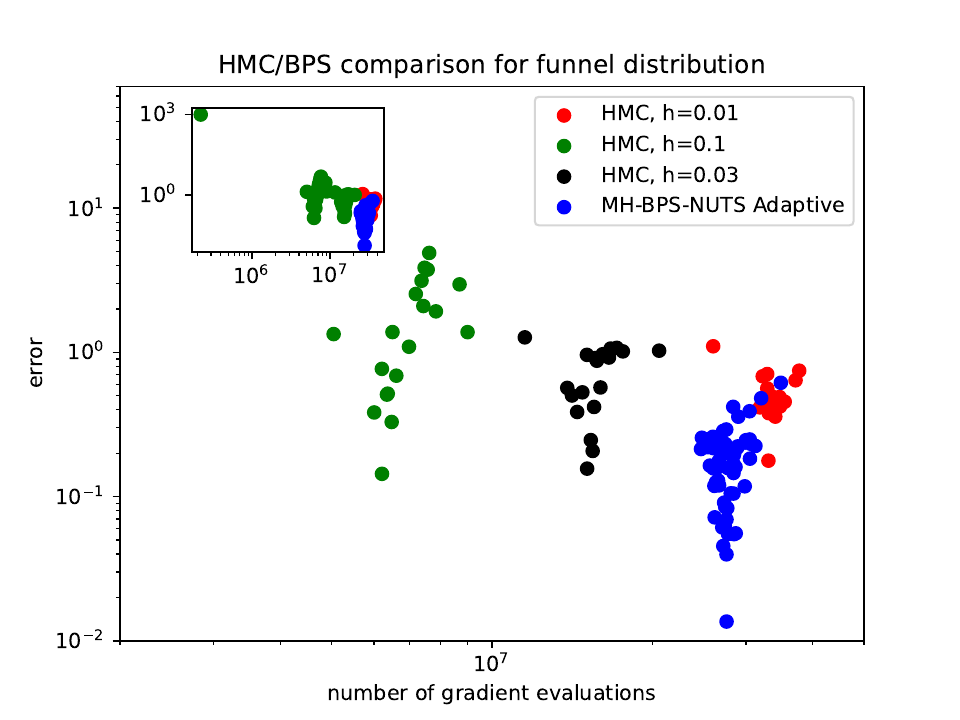}
    \caption{Comparison between HMC with different step sizes, and MHBPS-NUTS for the funnel.}
    \label{fig:funnel-comparison}
\end{figure}

\subsection{Simulated Tempering}
\label{subsec:simulated-tempering}

Simulated Tempering is an approach that constructs a path of distributions from a simple reference distribution $\pi_0$ to a complex target $\pi_1$, allowing the sampler to move between temperatures while targeting a single augmented distribution \cite{BironLattes2025,sutton2022,Tawn2020,Marinari1992}. A major attraction of this approach is that, when the reference distribution can be sampled independently, returns to the reference induce regenerative structure and hence offer a natural route to parallel computation via independent tours. This is one of the central motivations in recent work on non-reversible simulated tempering \cite{BironLattes2025}. However, a limitation of running a single sampler across all temperatures is that the kernel must remain effective over the entire path of distributions, making global tuning substantially more demanding. 

Our final example is designed to illustrate this difficulty for a tempering problem with a strongly varying geometry across temperatures. We consider a geometric path of distributions, $\pi_t = \pi_0^{1-t}\pi_1^tp(t)$, where the reference distribution, $\pi_0$, is a diffuse zero mean Gaussian with covariance $10I_2$, the target, $\pi_1$, is the Rosenbrock distribution \cite{Pagani2022} and $p(t)$ are the probabilities assigned to each $t$ (also known as the `affinities' \cite{BironLattes2025}). The sequence uses the power fraction schedule $t_i = (i/4)^5$ for $i=0,\dots,4$ \cite{Friel2008}. Temperature affinities were chosen so that the marginal distribution on temperature was approximately uniform.

We used the automatic PDMP sampler with the order-1 approximation to the event rate. As a benchmark, we compared against HMC-NUTS with three fixed step sizes, $\epsilon \in \{0.05,0.07,0.10\}$. Each method was run for $10^5$ iterations and the experiment was repeated five times.

\begin{figure}[!h]
    \centering
    \includegraphics[width=0.8\linewidth]{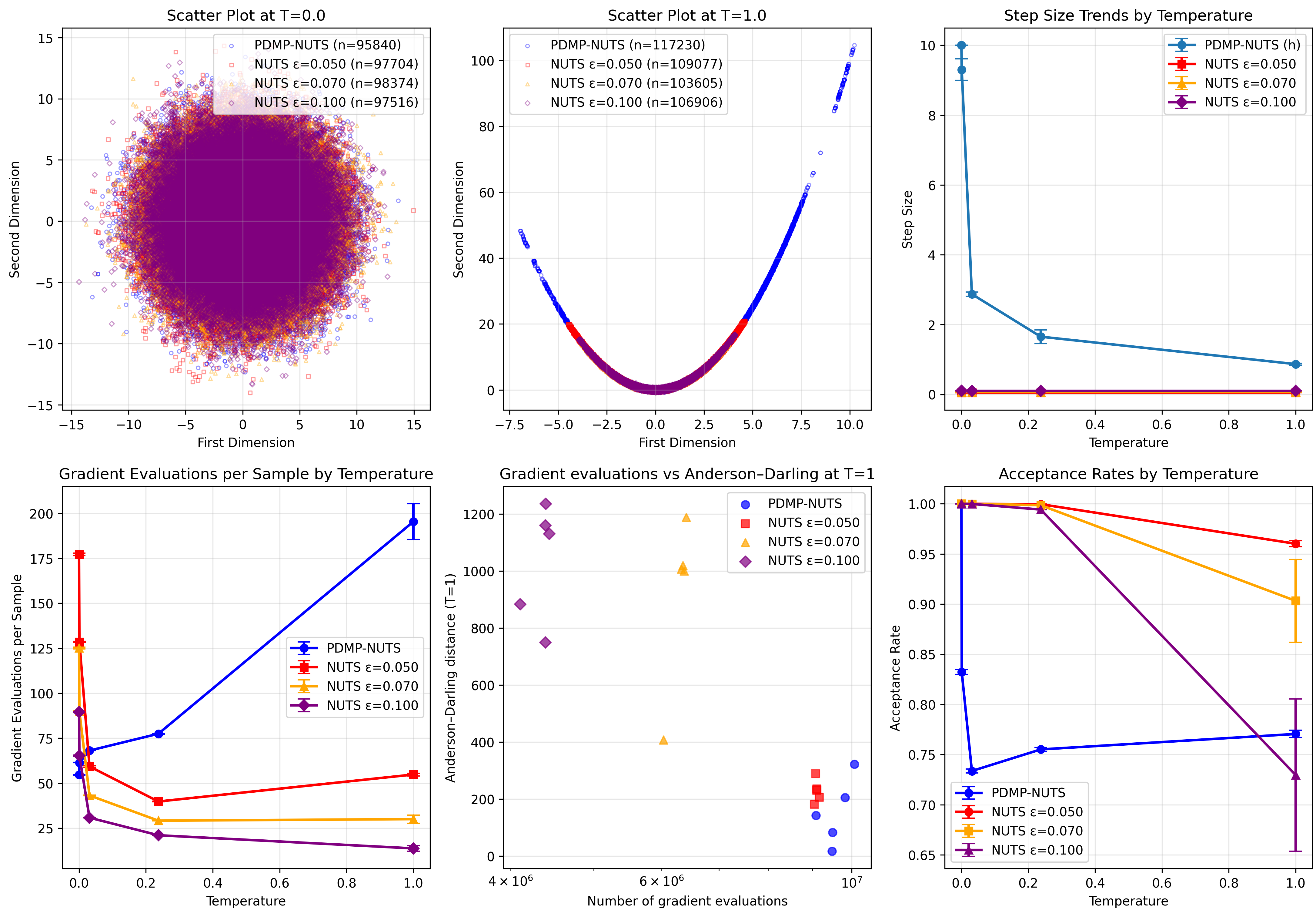}
    \caption{Comparison between HMC with different step sizes, and order-1 MHBPS-NUTS for the simulated tempering example.}
    \label{fig:tempering-comparison}
\end{figure}

Figure \ref{fig:tempering-comparison} shows a clear difference in robustness across the temperature path. At the reference distribution, all methods produce samples consistent with the broad Gaussian reference. However, when conditioned on being at the target, $t=1$, the HMC-NUTS samplers exhibit visibly poorer exploration of the tail regions, whereas the PDMP-based method reaches substantially further along the curved support of the target. This discrepancy is also reflected quantitatively by the Anderson--Darling distance at the target distribution. Despite operating at a computational budget comparable to, or in some cases smaller than, the HMC competitors, the automatic PDMP method achieves markedly better distributional fidelity in the tails. We see from the figure that the inability of HMC-NUTS to adapt the step size to the ideal temperature results in a low $\epsilon$ required to effectively sample the target, and consequently a large computational budget for the lower temperatures for HMC. The PDMP sampler in comparison attains stable acceptance at about $80\%$ (except at the reference distribution, where the order 1 approximation is exact) and the computation is not too large for the early temperatures.

\newpage

%%%%%%%%%%%%%%%%%%%%%%%%%%%%%%%%%%%%%%%%%%%%%%%%%%%%%%%%%%%%%%%%%%%%%%%%%%%%%%%%%%%%%%%%%%%%%%%%%%%%%%%%%%%%%%%%%%%%%
\section{Discussion}
\label{sec:disc}

In this work, we have introduced a novel algorithm that advances the practicality of PDMP-based methods for generic MCMC sampling problems. Our approach eliminates the need for a priori bounds on event rates while maintaining asymptotic exactness, thereby overcoming a significant limitation of existing PDMP methods. 

Beyond this key advantage, our algorithm allows for local adaptivity with respect to both discretisation in time and the duration of simulated trajectories, providing increased flexibility in exploring complex probability distributions. Empirical results demonstrate that it outperforms HMC on challenging distributions such as Neal’s funnel. Furthermore, adaptivity is more naturally incorporated in our framework compared to HMC: while adaptive step-size algorithms for HMC must keep the step size fixed within each iteration, our approach allows fully adaptive step-size selection at all times.

Building on the theory developed here, recent work has already successfully introduced strategies for local mass matrix adaptation in PDMPs \cite{chevallier2025covariance}. The PDMP constructed in that subsequent work is highly complex, making it practically impossible to derive viable a priori bounds. This demonstrates how eliminating the need for such bounds directly facilitates the development of more advanced samplers.

While our results highlight the practical benefits of PDMPs, several open challenges remain. One key limitation is the computational complexity of our adaptation of the No-U-Turn criterion, which currently scales as $O(n^2)$ per iteration with $n$ the number of events, which is significantly worse than the complexity of the corresponding criterion in HMC. Developing more efficient stopping rules remains an important direction for future research. Finally, while our method provides clear advantages in certain regimes, PDMPs generally exhibit slightly worse scaling with dimension compared to HMC. Bridging this gap requires the development of fundamentally improved PDMP-based samplers.

Overall, our work provides an important step towards making PDMPs a practical and competitive alternative to existing MCMC methods. Continued research into their computational efficiency and scalability will be crucial for broader adoption in Bayesian inference and related applications.

% % \textbf{[ SP: the $n^2$ comes from comparing all pairs of events. NUTS may prune some of the pairs (of endpoints) which can be ignored.]} \ac{not sure if there was anything to add to this?}
% % \textbf{[ SP: context for the mass matrix adaptation is the adaptivity angle; we are not adaptive in this way.]} \ac{paper on that has been published, so I've added some parts about this}

% \begin{comment}
% \ac{Prompt used to generate the conclusion:
% I'm writing the conclusion for my paper. Can you write a conclusion in Latex, in an acadamic style, not too fancy, with those points:
% This paper:
% 1. Our new algorithm propose a solution for using PDMP in practice and no longer requires bounds. 
% 2. It has the added benefit of being locally  adaptive in both time step and path length
% 3. It works better than HMC for funnel distribution
% 4. Adaptivity is more natural in this framework than for HMC. Algo that adapts the time step in HMC cannot adapt them for every leap frog iteration. Our algo has no such limitation
% Future work:
% 1. Improve the no u turn criterion. Currently requires O(n²) complexity to evaluate, while hmc one is much better
% 2. Mass matrix adaptivity?
% 3. PDMP scaling with dimension is slightly worse than HMC. This is a gap that cannot be bridged unless beter based PDMP are found.}
% \end{comment}
\newpage

\section{Acknowledgments}
\label{sec:Acknowledgments}
Matthew Sutton acknowledges the financial support from the Australian Research Council under the grant DE250101447.

%%%%%%%%%%%%%%%%%%%%%%%%%%%%%%%%%%%%%%%%%%%%%%%%%%%%%%%%%%%%%%%%%%%%%%%%%%%%%%%%%%%%%%%%%%%%%%%%%%%%%%%%%%%%%%%%%%%%%
\bibliography{main}
\bibliographystyle{ieeetr}
\newpage

%%%%%%%%%%%%%%%%%%%%%%%%%%%%%%%%%%%%%%%%%%%%%%%%%%%%%%%%%%%%%%%%%%%%%%%%%%%%%%%%%%%%%%%%%%%%%%%%%%%%%%%%%%%%%%%%%%%%%
\appendix
\section{Order 1 adaptation}
\label{appendix:order-1-adaptation}

Here we present a rule for the order-1 approximation, satisfying Assumption 2. For order-1 consider the Talyor's expansion of $\lambda(s) = \max(0, f(s))$ where, 
\begin{align*}
f(s) &= f(0)+f'(0)s+\frac{1}{2}f''(0)s^2+\frac{1}{6}f'''(0)s^3+O(s^4), \\ 
I_h &= \int_0^hf(s)ds = f(0)h+\frac{1}{2}f'(0)h^2+\frac{1}{6}f''(0)h^3+O(h^4).
\end{align*}
The order 1 approximation is $\bar{f}(s) = f(\left\lfloor s\right\rfloor_h) + \frac{1}{h}(s-\left\lfloor s\right\rfloor_h)(f(\left\lfloor s\right\rfloor_h + h) - f(\left\lfloor s\right\rfloor_h))$. For $s<h$ we have $\bar{f}(s) = f(0)+ \frac{s}{h}f(h)-\frac{s}{h}f(0)$ and the approximate integral is,
\begin{align*}
\bar{I}_h &= \int_0^h\bar{f}(s)ds = \frac{h}{2}(f(h) + f(0))    \\
&= hf(0)+\frac{1}{2}f'(0)h^2+\frac{1}{4}f''(0)h^3+O(h^4).
\end{align*}
The approximation error is thus $\bar{I}_h - I_h = \frac{f''(0)}{12}h^3 + \mathcal{O}(h^4)$. By defining the error coefficient $K = \frac{f''(0)}{12}$, we seek a step size $h = (\text{tol}/K)^{1/3}$ to bound the error by a given tolerance.

To estimate $K$, we compare a single step of size $h$ ($\bar{I}_{h,1}$) with two steps of size $h/2$ ($\bar{I}_{h/2,2}$):
\begin{align*}
    \bar{I}_{h,1} &= I_h + K h^3 + \mathcal{O}(h^4) \\
    \bar{I}_{h/2,2} &= I_h + 2 K (h/2)^3 + \mathcal{O}(h^4) = I_h + \frac{K}{4}h^3 + \mathcal{O}(h^4)
\end{align*}
Subtracting these yields $\bar{I}_{h,1} - \bar{I}_{h/2,2} = \frac{3}{4} K h^3 + \mathcal{O}(h^4)$. We therefore estimate $K$ as $\hat{K} = \frac{4}{3h^3}(\bar{I}_{h,1} - \bar{I}_{h/2,2})$, which satisfies $\hat{K} = K + \mathcal{O}(h)$.

The adapted step size $h$ is computed as:
\begin{equation}
    h = h_{guess} \left( \frac{3 \cdot \text{tol}}{4 |\bar{I}_{h_{guess},1} - \bar{I}_{h_{guess}/2,2}|} \right)^{1/3}
\end{equation}
where $\bar{I}_{h_{guess},1} = \frac{h_{guess}}{2}(f(h_{guess}) + f(0))$ and $\bar{I}_{h_{guess}/2,2} = \frac{h_{guess}}{4}(f(h_{guess}) + 2f(h_{guess}/2) + f(0))$.

\paragraph{proof of Proposition \ref{prop:scale-invariance-order-1} (scale invariance)}

For the scaled target, the rate is $\lambda_\sigma(t) = \sigma \lambda(\sigma t)$, implying $f_\sigma(t) = \sigma f(\sigma t)$. The second derivative satisfies $f''_\sigma(0) = \sigma^3 f''(0)$. 
The true error coefficient for the scaled process is $K_\sigma = \frac{1}{12}f''_\sigma(0) = \sigma^3 K$, where $K$ is the coefficient for the unscaled target.

Applying the estimation formula to the scaled process, we obtain:
\[
    \hat{K}_\sigma = K_\sigma + \mathcal{O}(h_{guess}) = \sigma^3 K + \mathcal{O}(h_{guess}).
\]
The unscaled empirical step size satisfies $h = (\text{tol}/|\hat{K}|)^{1/3} = (\text{tol}/|K|)^{1/3}(1 + \mathcal{O}(h_{guess}))$. Substituting $\hat{K}_\sigma$ into the scaled step size formula $h_\sigma = (\text{tol}/|\hat{K}_\sigma|)^{1/3}$ yields:
\[
    h_\sigma = \left( \frac{\text{tol}}{|\sigma^3 K + \mathcal{O}(h_{guess})|} \right)^{1/3} = \frac{1}{\sigma} \left( \frac{\text{tol}}{|K|} \right)^{1/3} \left( 1 + \mathcal{O}(h_{guess}) \right) = \frac{1}{\sigma} h \left( 1 + \mathcal{O}(h_{guess}) \right).
\]

\section{Technical lemma}

\begin{lemma}
    \label{lemma:technical-vol-change}
    Let $E$ and $F$ be two spaces. Let $f$ be defined as
    \begin{align*}
        f: \mathbb{R}^2 \times E &\rightarrow \mathbb{R}^2 \times F \\
        (x,t,x_E) &\mapsto (x - t,-t,f_E(x_E)),
    \end{align*}
    and 
    \begin{align*}
        h: \mathbb{R} \times E &\rightarrow \{0\} \times \mathbb{R} \times F \\
        (x,x_E) &\mapsto f(x,x,x_E).
    \end{align*}
    Furthermore, assume that $f$ has a change of volume $\Psi$, that is for every continuous function $\varphi$ with compact support, 
    \[\int_{\mathbb{R}^2 \times E} \varphi(f(x,t,x_E)) dxdtdx_E = \int_{\mathbb{R}^2 \times F} \varphi(y_1,y_2,y_F) \Psi(y_1,y_2,y_F) dy_1dy_2dy_F. \]
    Also assume that $f$ and $\Psi$ are continuous in $t$.
    Then the volume change associated to $h$ has density $\Psi(0,y_2,y_F)$.
\end{lemma}
\begin{proof}
    With $\varphi(y_1,y_2,y_F) = \varphi(y_1) \varphi(y_2) \varphi(y_F)$:
    \begin{align*}
        \int_{\mathbb{R}^2 \times E} \varphi_1(x-t) \varphi_2(-t) \varphi_3(f_E(x,t,x_E) dxdtdx_E &= \int_{\mathbb{R}^2 \times F} \varphi(y_1,y_2,y_F) \Psi(y_1,y_2,y_F) dy_1dy_2dy_F \\
        &= \int_{\mathbb{R}^2 \times F} \varphi_1(y_1) \varphi_2(y_2) \varphi_3(y_F) \Psi(y_1,y_2,y_F) dy_1dy_2dy_F \\
        &= \int_{\mathbb{R}^2 \times F} \varphi_1(y_1) \varphi_2(-t) \varphi_3(y_F) \Psi(y_1,-t,y_F) dy_1dtdy_F \\
        &= \int_{\mathbb{R}^2 \times F} \varphi_1(x-t) \varphi_2(-t) \varphi_3(y_F) \Psi(x-t,-t,y_F) dxdtdy_F.
    \end{align*}

    We write $\phi(x,t) = \int_E \varphi_3(f_E(x,t,x_E) dx_E$ and $\Phi(x,t) = \int_F \varphi_3(y_F) \Psi(x-t,-t,y_F) dy_F$. Hence, almost surely on $\mathbb{R}^2$ we have that $\phi(x,t) = \Phi(x,t)$. 
    By continuity on $t$, for almost every $x$, $\phi(x,t) = \Phi(x,t)$ for all $t$. Hence, for almost every $x$, $\phi(x,x) = \Phi(x,x)$.

    \begin{align*}
        \int_{\mathbb{R}\times E} \varphi(h(x,x_E)) dxdx_E&= \int_{\mathbb{R}} \varphi_1(0) \varphi_2(-x) \phi(x,x) dx\\
        &= \int_{\mathbb{R}} \varphi_1(0) \varphi_2(-x) \Phi(x,x) dx \\
        &= \int_{\mathbb{R} \times F} \varphi_1(0) \varphi_2(-x) \varphi_3(y_F) \Psi(0,-x,y_F) dx dy_F \\
        &= \int_{\mathbb{R} \times F} \varphi_1(0) \varphi_2(x) \varphi_3(y_F) \Psi(0,x,y_F) dx dy_F,
    \end{align*}
    which concludes the proof.
\end{proof}

\section{Proof of proposition \ref{prop:law_PlX}}
\label{sec:nuts-proof}

For a given $a < 0 < b$, we write $W_{[a,b]}$ for the set of paths on $[a,b]$. This set is:
\[
    W_{[a,b]} = E \times \cup_k (\mathbb{R^-} \times V^r)^k \times \cup_k (\mathbb{R^+} \times V)^k,
\]
which corresponds to the initial point at time $0$, then all the events backward in time, and then all the events forward in time. Note that unlike in the previous section, backward events are negative. 

In Section \ref{sec:metropolis}, we defined $R$ as the application that reversed a path  of length $T$. We shall call this application $R_T$ and the associated change of volume $\psi_T$. We will also consider that this application flips the signs of the jump times, unlike in Section \ref{sec:metropolis}.

\begin{proposition}[probability on path space]
\begin{enumerate}
    \item The PDMP defines a probability density $p_{[a,b]}$ on $W_{[a,b]}$ with respect to the canonical measure of $W_{[a,b]}$, and 
    \item the conditional probabilities $p((z_t)_{t\in [a,0)} | z_0)$ and $p((z_t)_{t \in (0,b]} | z_0)$ have densities on the spaces $\{z_0\} \times \cup_k (\mathbb{R^-} \times V^r)^k$ and $\{z_0\} \times  \cup_k (\mathbb{R^+} \times V)^k$ respectively.
\end{enumerate}
\end{proposition}

To avoid confusion, we denote $q((z_t)_{[0,b]} | x) = p((z_t)_{[0,b]}| z_0 = x)$ the conditional probability of a forward path starting at state $x$, and $q^r((z_t)_{[0,-a]}|x) = p((z_t)_{[a,0]}|z_a = x)$ the probability of a backward path starting at $x$.

%note: the following is true but useless, therefore i'm commenting it

Let us define the set of paths stopped on the left as
\[
    F_g = \left(E\times \cup_{k_b,k_f} (\mathbb{R}\times V^r)^{k_b}\times (\mathbb{R}\times V)^{+k_f} \right) \times \mathbb{R},
\]
the set of paths stopped on the right as
\[
    F_d = \mathbb{R} \times \left(E\times \cup_{k_b,k_f} (\mathbb{R}\times V^r)^{k_b}\times (\mathbb{R}\times V)^{+k_f} \right),
\]
and finally the set of stopped paths as $F = F_g \cup F_d$ .

Note that we can decide that the stopping criterion stops on paths of length greater than some $M > 0$, which can be arbitrarily big. Hence it is natural to look at the space $W_M = W_{[-M,M]}$ of paths on $[-M,M]$, which has density written as $p_M$ for simplicity.

Let $f :W_M \times [0,1] \mapsto F$ be the application that maps a given path $w \in W_M$ and $\alpha \in [0,1]$ to the stopped path given by the algorithm. 

A given path $w \in W_M$ can be written as $(z, t_1^b,v^b_1...,t_{k_b}^b,v_{k_b}^b,t_1^f,v_1^f,...,t_{k_f}^f,v_{k_f}^f)$.

\begin{lemma}[probability density on $F_g$]
\label{lemma:F_g-measure}
Let $S_g: F_g \mapsto \{0,1\}$ be the function that is 1 if the path was actually stopped on the left, and is $0$ otherwise. Then, the stopped PDMP induces a probability density on $F_g$ with density:
%\[
%    p_{F_g}(z,t^b_{k_1},v^b_{k_1},...,t^f_{k_2},v^f_{k_2},b) = p_{[0,b-t_{k_1}^b]}((z_{t-t_{k_1}^b})_{t \in [0,t_{k_1}^b]}) \left|\frac{-t^b_{k_1}}{(b-t_{k_1}^b)^2} \right| S_g(z,t^b_{k_1},...,t^f_{k_2},b).
%\]
\[
    p_{F_g}(z,t^b_{k_1},v^b_{k_1},...,t^f_{k_2},v^f_{k_2},b) = p_{W_{[t^b_{k_1},b]}}(z,t^b_{k_1},v^b_{k_1},...,t^f_{k_2},v^f_{k_2}) \left|\frac{-t^b_{k_1}}{(b-t_{k_1}^b)^2} \right| S_g(z,t^b_{k_1},v^b_{k_1},...,t^f_{k_2},v^f_{k_2},b).
\]
\end{lemma}
\begin{proof}
For a given $\alpha \in [0,1]$, assuming that the path is stopped on the left, there exists $k_1\leq k_b,k_2\leq k_f$ such that
\[
    f(w,\alpha) = (z,t_1^b,...,t_{k_1}^b,t_1^f,...,t_{k_2}^f, (1-\alpha) \frac{-t^b_{k_1}}{\alpha})
\]

Let $A \subset F$ of the form $A \subset E \times  \mathbb{R}^{k_1+k_2} \times \mathbb{R} $, and let $P_F$ be the pushforward measure of $p_M \times \text{Unif}(0,1)$ by $f$.
By definition:
\[
P_F(A) = \int_{W_M} \int_0^1 1_A(f(w,\alpha)) p_M(w) d\alpha dw.
\]
We can rewrite this as:
\begin{align*}
    P_F(A) &= \sum_{k_b \geq k_1, k_f> k_2}  \int_{(\mathbb{R}\times V^r)^{t^b_{k_b} - k_1} \times (\mathbb{R}\times V)^{t_{k_f}^f - k_2 - 1}} \int_E \int_{(\mathbb{R}\times V^r)^{k_1} \times (\mathbb{R}\times V)^{k_2+1}} \int_0^1  \\
    &S_g(z,t^b_{k_1},v^b_{k_1},...,t^f_{k_2},v^f_{k_2},(1-\alpha) \frac{-t^b_{k_1}}{\alpha}) 1_A(z,t^b_{k_1},v^b_{k_1},...,t^f_{k_2},v^f_{k_2},(1-\alpha) \frac{-t^b_{k_1}}{\alpha}) \\
    &1_{t^f_{k_2+1} > (1-\alpha) \frac{-t^b_{k_1}}{\alpha}} p_M(...) d \alpha dz d_{t_{k_b}^b} d_{v_{k_b}^b} ... d_{t_{k_f}^f}  d_{v_{k_f}^f}
\end{align*}
The change of variables $b = (1-\alpha) \frac{-t^b_{k_1}}{\alpha}$ (i.e. $\alpha = -t_{k_1}^b/(b-t_{k_1}^b)$) for a fixed $t^b_{k_1}$ yields:
\begin{align*}
    P_F(A) &= \sum_{k_b \geq k_1, k_f> k_2}  \int_{(\mathbb{R}\times V^r)^{t^b_{k_b} - k_1} \times (\mathbb{R}\times V)^{t_{k_f}^f - k_2 - 1}} \int_E \int_{(\mathbb{R}\times V^r)^{k_1} \times (\mathbb{R}\times V)^{k_2+1}} \int_\mathbb{R} \\
    &S_g(z,t^b_{k_1},v^b_{k_1},...,t^f_{k_2},v^f_{k_2},b) 1_A(z,t^b_{k_1},v^b_{k_1},...,t^f_{k_2},v^f_{k_2},b) 1_{b< t_{k_2+1}^f} \\
    &p_M(...) \left|\frac{-t^b_{k_1}}{(-t_{k_1}^b/(b-t_{k_1}^b))^2} \right|^{-1} db dz d_{t_{k_b}^b} d_{v_{k_b}^b} ... d_{t_{k_f}^f}  d_{v_{k_f}^f} \\
    &= \int_E \int_{(\mathbb{R}\times V^r)^{k_1} \times (\mathbb{R}\times V)^{k_2}} \int_\mathbb{R} 
    S_g(z,t^b_{k_1},v^b_{k_1},...,t^f_{k_2},v^f_{k_2},b) 1_A(z,t^b_{k_1},v^b_{k_1},...,t^f_{k_2},v^f_{k_2},b) \\
    & p_{W_{[t^b_{k_1},b]}}(z,t^b_{k_1},v^b_{k_1},...,t^f_{k_2},v^f_{k_2}) \left|\frac{-t^b_{k_1}}{(-t_{k_1}^b/(b-t_{k_1}^b))^2} \right|^{-1} db dz d_{t_{k_1}^b} d_{v_{k_1}^b} ... d_{t_{k_2}^f} d_{v_{k_2}^f},
\end{align*}
which gives us the density on $F_g$:
\[
    p_{F_g}(z,t^b_{k_1},v^b_{k_1},...,t^f_{k_2},v^f_{k_2},b) = p_{W_{[t^b_{k_1},b]}}(z,t^b_{k_1},v^b_{k_1},...,t^f_{k_2},v^f_{k_2}) \left|\frac{-t^b_{k_1}}{(b-t_{k_1}^b)^2} \right| S_g(z,t^b_{k_1},v^b_{k_1},...,t^f_{k_2},v^f_{k_2},b).
\]
\end{proof}
The random variable $(X,l)$ has values in $E \times \cup_k (\mathbb{R^+} \times V)^k \times \mathbb{R}^2$, defined by a path forward starting at time $0$ in $E \times \cup_k (\mathbb{R^+} \times V)^k$, a total length in $\mathbb{R}$, and the starting time (i.e. $l$). However, the path was stopped at an event. Let us take $(z,t_0,v_0,...,t_k,v_k,T,l) \in E \times \cup_k (\mathbb{R^+} \times V)^k \times \mathbb{R}^2$. If the path was stopped backward, then $t_0 = 0$. Otherwise the path was stopped forward, and $z_k = T$. Hence, the state space of $(X,l)$ is:
\[
    F_X = \{ (z,t_0,v_0,...,t_k,v_k,T,l) \in E \times \cup_k (\mathbb{R^+} \times V)^k \times \mathbb{R}^2 | t_0 = 0 \text{ or } t_k = T\}.
\]
In the case $t_0 = 0$, it means that the coordinate associated to $t_0$ can be ignored. On this part of the space, we take the canonical measure defined on $E \times \cup_k \left( \{0\}\times V \times(\mathbb{R^+} \times V)^{k-1} \right) \times \mathbb{R}^2$. We do the same for the part of space where $t_k = T$.
This means that the dimension has been reduced by 1, which compensates for adding the variable $l$, meaning that $F_g$ (and $F_d$), have the same dimension as $F_X$.

Let $g$ be the transformation associating a path in $F_g$ to a path in $F_X$.  In other words, the pushforward measure of $p_{F_g}$ by $g$ is the law of $(X,l)$. The application $g$ applies $R_{t_{k_1}}$ to the reversed part of the path and shifts the rest by time $-t_{k_1}$:
\begin{align*}
    g: F_g & \mapsto E \times \cup_k \left( \{0\}\times V \times(\mathbb{R^+} \times V)^{k-1} \right) \times \mathbb{R}^2 \subset F_X\\
    (z,t^b_{k_1},v^b_{k_1},...,t^f_{k_2},v^f_{k_2},b) &\mapsto (R_{-t_{k_1}}(z,t^b_{k_1},v^b_{k_1},...,t^b_1,v^b_1),t^f_1 - t^b_{k_1},v^f_1,...,t^f_{k_2} - t^b_{k_1},v^f_{k_2},b - t^b_{k_1}, -t_{k_1}^b).
\end{align*}
We could define a similar application $d$ from $F_d$ to $F_X$, which would follow similar arguments.
\begin{lemma}[Volume change of $g$]
    \label{lemma:g-volume-change}
    Assume that for all paths, the mappings $t \mapsto R_t$ and $t \mapsto \psi_t$ are continuous. It then holds that the application $g$ has volume change for $(z,t_1,v_1,...,t_k,v_k,T,l) \in E \times \cup_k \left( (\mathbb{R^+} \times V)^{k} \right) \times \mathbb{R}^2$ given by
    \[
        \Psi_l(z,t_1,v_1,...,t_{k_l},v_{k_l}),
    \]
    where $k_l$ is the number of event times smaller than $l$.
\end{lemma}
\begin{proof}
    We define the function $f$ as follows:
    \begin{align*}
    f: F_g \times \mathbb{R} & \mapsto E \times \cup_k \left( (\mathbb{R^+} \times V)^{k} \right) \times \mathbb{R}^2\\
    (z,t^b_{k_1},v^b_{k_1},...,t^f_{k_2},v^f_{k_2},b,t) &\mapsto (R_{-t}(z,t^b_{k_1},v^b_{k_1},...,t^b_1,v^b_1),t^f_1 - t^b_{k_1},v^f_1,...,t^f_{k_2} - t^b_{k_1},v^f_{k_2},b - t^b_{k_1}, -t).
\end{align*}
    This goes into a bigger space than $F_X$. However for $t = t_{k_1}^b$, $f$ and $g$ are the same. The volume change associated to $f$ where $(z,t_1,v_1,...,t_k,v_k,T,l) \in E \times \cup_k \left( (\mathbb{R^+} \times V)^{k} \right) \times \mathbb{R}^2$ is
    \[
        \Psi_l(z,t_1,v_1,...,t_{k_l},v_{k_l}),
    \]
    where $k_l$ is the number of event times smaller than $l$. 
    The event time associated to $t_{k_b}^b$ by $R_{-t}$ is $t_{k_b}^b -t$, therefore we can apply Lemma \ref{lemma:technical-vol-change}, and we deduce that the volume change of $g$ is
    \[
        \Psi_l(z,0,v_1,...,t_{k_l},v_{k_l}).
    \]
\end{proof}
\begin{proposition}[conditional probability of $l|X$]
    \label{prop:nuts-proof-law-l}
    The conditional probability of $l$ knowing $X$ has a density on $\mathbb{R}$ which is
    \begin{align*}
        p(l | X) &\propto l 1_{l \in [0,T]} \text{ if the path is stopped on the left}, \\
        p(l | X) &\propto (T-l) 1_{l \in [0,T]} \text{ if the path is stopped on the right}, 
    \end{align*}
    where $T$ is the length of $X$.
\end{proposition}
\begin{proof}
We only analyze the $F_g$ case, as the $F_d$ case follow the same logic.
The law of the couple of random variables $(X,l)$ is the pushforward of the probability on $F_g$ (and $F_d)$ by $g$. 

For a couple $(X,l)$ that is built by the algorithm and stopped left (i.e. $S_g = 1$), we write the variable $X$ as a path on $[0,T]$. Using Lemma \ref{lemma:g-volume-change}:
\[
    p(X,l) = p_{F_g}(g^{-1}(X,l)) \Psi_l(X_{[0,l]}).
\]
Using Lemma \ref{lemma:F_g-measure}, we conclude that
\[
    p(X,l) = \pi(X_l) q^r_{[-l,0]}((X_{t+l})_{t\in [-l,0]}|X_l) q_{[0,T-l]}((X_{t+l})_{t\in [T-l,T]}|X_l) \left| \frac{l}{T^2} \right| \Psi_l(X_{[0,l]}).
\]
Using \eqref{eq:backward-density} and the Markov property:
\[
    p(X,l) = p_{[0,T]}((X_t)_{t\in[0,T]}) \left| \frac{l}{T^2} \right|.
\]
The first part $p_{[0,T]}((X_t)_{t\in [0,T]})$ depends only on $X$ and not on $l$. Furthermore, if $(X,l)$ is stopped on the left, $S_g = 1$ for any $l \in [0,T]$, and is $0$ otherwise. Hence:
\[
    p(l | X) \propto l 1_{l \in [0,T]}.
\]
\end{proof}

\end{document}